\documentclass[11pt]{article}

\usepackage{amssymb} 
\usepackage{latexsym} 
\usepackage{amsmath}
\usepackage{color}
\usepackage{epsfig}
\usepackage{graphicx}
\usepackage[all]{xy}
\input xy
\usepackage{enumerate}
\usepackage{anysize}
\usepackage[english]{babel}
\usepackage{bbm}
\usepackage{hyperref}

\newtheorem{remark}{Remark}[section]
\newtheorem{theorem}[remark]{Theorem}
\newtheorem{definition}[remark]{Definition}

\newtheorem{prop}[remark]{Proposition}

\newtheorem{lemma}[remark]{Lemma}

\newtheorem{coro}[remark]{Corollary}

\newtheorem{proofTh} {Proof.}

\newcommand{\CVD} {\hspace*{\fill}$\Box$}
\newcommand{\CVDC} {\hspace*{\fill}$\lozenge$}
\newenvironment{proof}{\begin{proofTh}\em}{\CVD\end{proofTh}}

\newtheorem{pclaimTh} {Proof.}

\newcommand{\mybreak} {\par\vspace{2mm}\noindent}
\newcommand{\mybreakk} {\par\vspace{1.5mm}\noindent}

\def\cadre{$$\vcenter\bgroup\advance\hsize by -2em\noindent
             \refstepcounter{equation}{\rm(\theequation)}~\ignorespaces}
\makeatletter
\def\endcadre{\egroup\eqno$$\global\@ignoretrue \par\vspace{0mm}\noindent}

\newcommand{\comment}[1]{}

\newcommand{\A} {\mathsf{A}}
\newcommand{\B} {\mathsf{B}}
\newcommand{\C} {\mathsf{C}}
\newcommand{\DD} {\mathsf{D}}
\newcommand{\E} {\mathsf{E}}
\newcommand{\HH} {\mathsf{H}}
\newcommand{\K} {\mathsf{K}}
\newcommand{\NN} {\mathsf{N}}

\newcommand{\TT} {\mathsf{T}}
\newcommand{\U} {\mathsf{U}}
\newcommand{\F} {\mathsf{F}}

\newcommand{\s} {\mathsf{S}}

\newcommand{\ra} {\mathsf{a}}
\newcommand{\rb} {\mathsf{b}}

\newcommand{\q} {\mathbf{q}}
\newcommand{\kk} {\mathbf{k}}

\newcommand{\uu} {\mathbf{u}}
\newcommand{\vv} {\mathbf{v}}
\newcommand{\x} {\mathbf{x}}
\newcommand{\y} {\mathbf{y}}
\newcommand{\w} {\mathbf{w}}
\newcommand{\z} {\mathbf{z}}
\newcommand{\tw} {\mathbf{t}}

\newcommand{\cv} {\mathbbm{v}}

\newcommand{\cw} {\mathbbm{w}}

\newcommand{\spt} {{\rm supp}}
\newcommand{\uno} {\mathsf{1}}

\makeatletter
\def\imod#1{\allowbreak\mkern10mu({\operator@font mod}\,\,#1)}
\makeatother
\title{Minimally unbalanced diamond-free graphs and Dyck-paths}

\author{Nicola Apollonio\footnote{Istituto per le Applicazioni del
Calcolo, M. Picone-CNR, via dei Taurini 19, 00185 Rome - Italy.
\texttt{nicola.apollonio@cnr.it}} \and {Anna Galluccio\footnote{Istituto di Analisi dei Sistemi ed Informatica "Antonio Ruberti"- CNR, Viale Manzoni, 30
00185 Rome - Italy. \texttt{anna.galluccio@iasi.cnr.it}}}}

\date{}
\begin{document}
\maketitle

\begin{abstract}
A $\{0,1\}$-matrix $\A$ is balanced if it does not contain a submatrix of odd order having exactly two 1's per row and per column. 
A graph is balanced if its clique-matrix is balanced. No characterization of minimally unbalanced graphs is known, and even
no conjecture on the structure of such graphs has been posed, contrarily to what happened for perfect graphs. 
In this paper, we provide such a characterization for the class of diamond-free graphs and establish a connection between minimally unbalanced diamond-free graphs and Dyck-paths.
\mybreak
\end{abstract}

\section{Introduction}\label{sec:intro}

A $\{0,1\}$-matrix $\A$ is \emph{balanced} if it does not contain a submatrix of odd order with two 1's per row and per column. 
This notion was introduced and thoroughly investigated by Berge \cite{Be72}. 

A $\{0,1\}$-matrix $\A$ is \emph{perfect} if the associated \emph{fractional packing polyhedra} 
$\big\{\x\in\mathbb{R}^n_+ \ |\ \A\x\leq \uno\big\}$, is an integral polytope, namely, it has integer vertices only.

Classical results of Berge, on one hand, and of Fulkerson, Hoffman and Oppenheim \cite{FHO74} on the other assert that $\A$ is balanced if and only 
if every submatrix of $\A$ is perfect. 

Perfect graphs and perfect matrices are related as follows: a graph is \emph{perfect} if and only if its clique-matrix is a perfect 
matrix. Recall that the clique-matrix $\A_G$ of a graph $G$ is a $\{0,1\}$-matrix whose columns are indexed by the vertices of $G$ and 
whose rows are incidence vectors of the maximal cliques of $G$. 

After the Strong Perfect Graph Theorem \cite{SPGT}, perfect graphs are characterized by a list of forbidden minimally imperfect graphs, 
namely non-perfect graphs all whose proper induced subgraphs are perfect: perfect graphs are precisely those graphs that do not 
contain induced odd holes or their complements as induced subgraphs. Equivalently, odd holes and their complements are the only 
minimally imperfect graphs.

In the same way as perfect graphs are those graphs whose clique-matrix is perfect, \emph{balanced} graphs are graphs whose clique-matrix 
is balanced. Moreover, since the clique-matrix of an induced subgraph $G'$ of a graph $G$ is a submatrix of $\A_G$, it follows that, 
like perfect graphs, the class of balanced graphs is closed under taking induced subgraphs. Therefore, it is natural to ask whether 
graph-balancedness can be characterized by a list of minimally forbidden subgraphs, that is \emph{minimally unbalanced graphs}, similarly 
to what happens for perfect graphs---a graph is \emph{minimally unbalanced} if it is not balanced but each of its proper induced 
subgraphs is balanced-. 

Even though there exists a polynomial-time algorithm to recognize balanced matrices based on a decomposition 
algorithm of Conforti, Cornuejols and Rao \cite{CCR99}, no such a characterization of minimally unbalanced graphs is known up-to-date and no 
conjecture has been formulated in this respect, contrarily to what happened for perfect graphs. 
A first attempt to characterize these obstructions was made by Bonomo et al. \cite{BDLS06} but the structures they identify, 
the {\em generalized odd suns}, though appliable to general graphs are far from being minimally unbalanced. 
\mybreak
In this paper we identify the complete list of minimally unbalanced graphs within the class of \emph{diamond-free graphs}, i.e., 
graphs with no induced copy of the diamond $K_4-e$, thereby giving a characterization of diamond-free balanced graphs by forbidden 
induced subgraphs. 

We focus on diamond-free graphs because their clique-matrix has the remarkable property of being \emph{linear}. 
A $\{0,1\}$-matrix is \emph{linear} if it does not contain ${1\,1\brack 1\,1}$ as a submatrix. 
A polynomial-time algorithm to recognize linear matrices was developed in a series of papers by Conforti and Rao \cite{CR93}
already in the late eighties, but no characterization of minimally non-balanced linear matrices is known. 
This because the algorithm relies on a decomposition and does not hint at the structure of the obstructions to 
balancedness. 
Now, as observed in Section~\ref{sec:notation} and in \cite{CR92b}, linear balanced matrices and diamond-free balanced graphs are essentially 
the same thing, so our characterization provides these obstructions to balancedness for linear matrices and, at the same time, allows a graph-theoretical 
interpretation of the algorithm given in \cite{CR92a} when specialized to linear matrices.

\mybreak
The complete characterization of minimally unbalanced diamond-free graphs is obtained in several steps that exploit different combinatorial 
constructions.
In Section~\ref{sec:multisun}, we exploit the properties of the linear clique-matrices to state that
minimally unbalanced diamond-free graphs belong to exactly two classes of graphs: 
odd holes and graphs that suitably generalize odd suns and that hereditarily satisfy  the property of being odd hole free (the 
{\em HOH-free multisuns}).
Unfortunately, this characterization does not say much about the structure of these graphs and we need techniques from other 
fields of combinatorics to provide a complete description of HOH-free multisuns. 
 
To this aim, we first identify a number of necessary conditions (the {\em N-conditions}) that are satisfied by HOH-free multisuns.
To formally handle these conditions, we associate words over a finite alphabet to multisuns. More precisely, 
in Section~\ref{sec:sword}, we prove that some equivalence class of words ({\em $s$-words}) are in one-to-one correspondence with 
families of graphs ${\cal S}_G$ consisting of even subdivisions/contractions of a multisun $G$ that satisfies the N-conditions.  

Unfortunately, the N-conditions are not sufficient to guarantee the HOH-freeness of a multisun $G$. So, in Section~\ref{sec:sunword},
we introduce the notion of {\em sunoid}, i.e., a multisun that satisfies the N-conditions as well as it does each of its proper sub-multisuns. 

Sunoids exhibit a large amount of geometrical structure and can be roughly described as the solution of the following combinatorial problem:
\mybreak
\emph{Take a collection of $p$ edge-disjoint cliques having one vertex in common. How to inscribe such cliques in an odd cycle $C$ in 
such a way that the resulting graph $G$ is diamond free and $G$ has no odd holes and so does any subgraph obtained by removing the edge-set of any $h\leq p-1$ cliques 
among the inscribed ones?} 
\mybreak

The simplest example of sunoid arises when $p=1$. In this case it suffices to place the vertices of one clique $K$ on $C$ so 
that two consecutive vertices of $K$ on $C$ are separated by a positive even number of vertices of $V(C)-K$ (see Figure~\ref{fig:1}). 
The term sunoid is due to the fact that when $p=1$ these graphs are subdivisions of an {\em odd sun}. 
\begin{figure}
    \begin{center}
             \includegraphics[width=4cm]{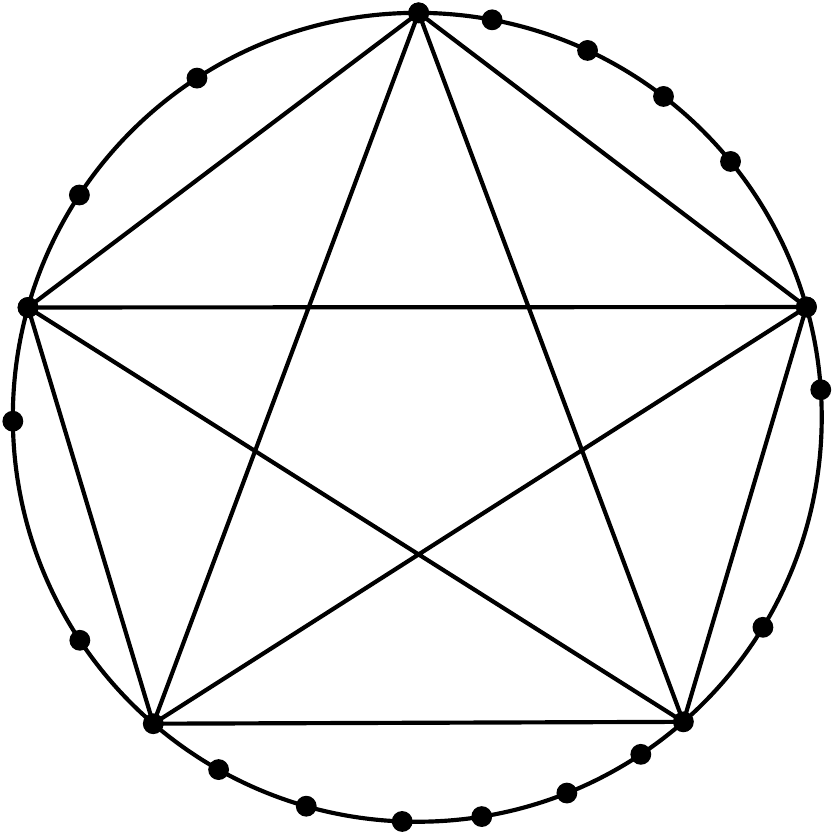}
    \end{center}
    \caption{A sunoid with one inscribed clique.}
    \protect\label{fig:1}
\noindent\hrulefill%
\end{figure}

Since sunoids form a subclass of multisuns, we represent them with special $s$-words: the ({\em sunwords}). 
This allows us to translate the geometrical structure of sunoids into two simple combinatorial conditions on $s$-words and to provide a
good characterization of sunoids, i.e., the membership problem for sunoids is in NP$\cap$Co-NP. 
This is described in details in Section~\ref{sec:sunword}. 

Finally, in Section~\ref{sec:mnb}, we show that sunoids are precisely the HOH-free multisuns, thus proving that they are, together with 
odd holes, the only obstructions to balancedness in diamond-free graphs.
+
In the last section we present some of the consequences of our result in apparently distant fields of combinatorics. 
Indeed, we observe that sunoids are intimately related with other well known combinatorial objects: the {\em Dyck-paths} \cite{De99}. 
Surprisingly enough, we prove that Dyck-paths are in correspondence with minimally unbalanced diamond-free graphs and 
this makes it possible  to enumerate them. 

\noindent
Less surprisingly there is a relationship between balancedness of graphs and another graph property known as \emph{clique-perfection}.
Indeed,  these two properties turns out to be equivalent in the class of diamond-free graphs 
(a simple proof of this fact is given in Section~\ref{sec:consequence}). 
In \cite{BCDII07} Bonomo et al. posed the following problem: 
{\it Is it possible to characterize diamond-free clique-perfect graphs in terms of minimally forbidden induced subgraphs?} 
This problem is solved in this paper because our characterization states that a diamond-free graph is 
clique-perfect if and only if it does not contain an odd hole or a sunoid as an induced subgraph.

\section{Definitions and basic facts}
\label{sec:notation}

Notation and terminology used throughout the paper is mostly standard. For $n\in\mathbb{N}$, $[n]$ is the set $\{1,\ldots,n\}$. The cardinality 
of a set $A$ is denoted by $\#A$. The order and the size of a graph are the cardinalities of its vertex- and edge-set, respectively. 
A \emph{clique} in a graph $G$ is a set of pairwise adjacent vertices and a \emph{stable set} is a set of pairwise non-adjacent vertices.
We do not distinguish between cliques of a graph $G$ and the subgraph they induce in $G$. A \emph{cycle} is a graph is a copy of 
$C_n$ while a hole is an induced copy of 
$C_n$ for $n\geq 4$. The cycle is \emph{even} or \emph{odd} according to the parity of $n$. In general we say that a finite set 
$A$ is is \emph{even} or \emph{odd} according to the parity of $\#A$. 

If $\A$ and $\B$ are $\{0,1\}$-matrices with the same number of rows and columns, then we write $\A\cong \B$ and say that $\A$ and $\B$ 
are \emph{congruent}, whenever $\A$ can be obtained from $\B$ by permuting its rows and columns.
The edge-vertex adjacency matrix of a cycle of length $n$ is referred to as an \emph{odd cycle matrix} (\emph{of order $n$}). 
Let $\C_n$ be the $\{0,1\}$-matrix matrix defined by $\C_n=(c_{i,j})$ where $c_{i,j}=1$ if $j=i,i+1$ and addition over indices is taken 
modulo $n$. Clearly any odd cycle matrix of order $n$ is congruent to $\C_n$. In particular any matrix congruent to $\C_3$ will be 
referred to as a \emph{triangle matrix}.

Let $\A$ be a $\{0,1\}$-matrix. 
A row $\ra$ of $\A$ is \emph{dominated} if there is some row $\rb$ of $\A$ such that $\ra\leq \rb$ componentwise. 
Otherwise row $\ra$ is \emph{maximal}. 
A \emph{row (column) submatrix} of $\A$ is a submatrix of $\A$ consisting of some rows (columns) of $A$. 

The {\em up-matrix} $\A^\uparrow$ is the row submatrix of $\A$ consisting of the maximal rows of $\A$,
i.e., the set on non-dominated rows of $\A$. Let $g(\A)\in \mathbb{N}\cup\{\infty\}$ be defined as follows: 
if $\A$ does not contain any odd cycle matrix as a submatrix, then $g(\A)=\infty$ otherwise $g(\A)$ is the least order of an 
odd cycle submatrix of $\A$.  Clearly $\A$ is balanced if and only if $g(\A)=\infty$.

The {\em intersection graph} of a matrix $\A$ is the graph $G_{\A}$ whose vertices are labelled by the columns of $\A$ 
and two vertices are adjacent if the corresponding columns are non-orthogonal. A $\{0,1\}$-matrix $\A$ is a {\em clique-matrix}
if $\A\cong \A_G$ for some graph $G$.
\mybreak
>From the definitions of clique-matrix and up-matrix it follows straightforwardly that:

\begin{lemma}
\label{lemma:2sec}
If $G$ is a graph and $G'$ is an induced subgraph of $G$, then $\A_{G'}$ is the up-matrix of the column submatrix of $\A_G$ 
consisting of the columns indexed by $V(G')$. Conversely, if $\A'$ is an up-matrix of $\A_G$, then $\A'\cong \A_{G'}$ for some induced 
subgraph $G'$ of $G$.
\end{lemma}
 
Matrix $\A$ is {\em conformal} if $\A_{G_\A}\cong\A^\uparrow$ (clearly, clique-matrices are always conformal). 
\emph{Gilmore's criterion of conformality} asserts that $\A$ is conformal if and only if whenever $\C\cong \C_3$ is a submatrix of $\A$ 
then ${\C \brack 1\, 1\, 1}$ is also a submatrix of $\A$. The following fact (whose proof is just a metter of checking definitions) 
establishes the link between linear matrices and clique-matrices of diamond-free graphs.

\begin{lemma}
\label{lemma:AB}
If $\A$ is a linear matrix, then either $g(\A)=3$ or $\A^\uparrow$ is conformal and it is the 
clique-matrix of a diamond-free graph $G$---take $G\cong G_\A$ and recall that $G_\A\cong G_{\A^\uparrow}$--. 
Conversely, if $G$ is a diamond-free graph, then $\A_G$ is a conformal linear matrix with $g(\A_G)\geq 5$.
\end{lemma}

The main device we employ in our characterization is a construction that associates a labeled cycle with a word over the set 
of its labels. The next two subsections recall some basic terminology on words and the reader can skip them  until  
Section~\ref{sec:sword} where these concepts are used for the first time.

\subsection{Linear words}

Words on a finite alphabet $\Sigma$ are finite sequences of elements of $\Sigma$ and will be denoted by boldface lowercase letters. 
The set of word on $\Sigma$ 
is denoted, as customary, by $\Sigma^*$. If $\w=w_1w_2\cdots w_n$, $w_i\in \Sigma$, $\forall i\in [n]$ then $n$ is the \emph{length} of the word. 
The \emph{support} $\spt(\w)$ of a word $\w$ is the set of symbols 
occurring in $\w$. The \emph{concatenation} of the words $\uu=u_1u_2\cdots u_m$ and $\vv=v_1v_2\cdots v_n$ is the word 
$\uu\vv=u_1u_2 \cdots u_mv_1v_2\cdots v_n$. An \emph{interval} in a word 
$\w$ is a word $\vv$ of $\Sigma^*$ such that $\w=\uu\vv\z$ for some other (possibly empty) two words $\uu,\,\z\in \Sigma^*$. 
When $\z=\phi$ we refer to $\vv$ is a {\em postfix} of $\w$. When $\uu=\phi$, $\vv$ is a {\em prefix} of 
$\w$. A \emph{subword} of $\w$ is the word obtained from $\w$ by setting to $\phi$ (i.e., by deleting) some of the letters of $\w$. 
Note that a subword of $\w$ is not necessarily an interval of $\w$. 
For instance $adc$ is a subword but not an interval of $aabddc$.  
\mybreak
Concatenation of sequences turns $\Sigma^*$ into a commutative monoid (the free monoid), with the empty string $\phi$ acting as 
neutral element. This structure makes it possible to use the exponential notation for words: if $x$ is a letter of $\Sigma$, then 
one writes $x^n$ for the word $\w=w_1w_2\cdots w_n$ such that $w_i=x$, $\forall i\in [n]$ and then one formally introduces the rule of 
exponents $x^hx^k=x^{h+k}$, for $h,\,k\in \mathbb{N}$, $x\in \Sigma$. A word might be presented in several ways using the exponential 
notation. Among them we distinguish the \emph{standard form}: $\w=w_{i_1}^{k_1}w_{i_2}^{k_2}\ldots w_{i_s}^{k_s}$ where 
$w_{i_j}\in \Sigma$, $j\in [s]$ and $w_{i_j}\not=w_{i_{j+1}}$ $j=1,\ldots,s-1$, i.e., powers occur with highest possible exponent. 

If $\w$ is presented in standard form as $\w=w_{i_1}^{k_1}w_{i_2}^{k_2}\ldots w_{i_s}^{k_s}$ and $\uu$ is 
an interval of $\w$, then we write $\uu\lhd \w$ if $\uu=w_{i_l}^{k_l}\ldots w_{i_m}^{k_m}$ for some $l$ and $m$ such that 
$1\leq l\leq m\leq s$, that is $\uu\lhd \w$ if the letters in $\uu$ appear with the same exponents as in $\w$. 
Notice that $\uu$ can be an interval of $\w$ such that $\uu\ntriangleleft\w$: for instance, if $\w=aaabddccd$, then, in standard form,  
$\w=a^3bd^2c^2d$; now $a^2bd^2c$ is the standard form of the interval $aabddc$ of $\w$ though $a^2bd^2c\ntriangleleft\w$. 
On the other hand, for instance, $a^3bd^2\lhd \w$.

\subsection{Cyclic words}
Let $C\cong C_n$. A labeling of $C$ is a surjective mapping $g: V(C)\rightarrow \Sigma$, where $\Sigma$ is a finite set.
The labeling $g$ induces a linear word on $\Sigma$ that depends on the starting point and the chosen orientation of $C$.
In principle, any automorphism of $C$ produces a different ``linear word'' that encodes however the same information. 
Therefore two different words $\vv$ and $\w$ associated with $G$ have to be considered the same if they are in the same orbit under the action of the dihedral group $D_n$ (the automorphism group of $C$).
In this case we write $\vv\sim \w$. In other words $\vv\sim \w$ if $\w$ can be made coincident with $\vv$ by shifting the 
indices, reversing the order of reading and by composing these actions.

Now we introduce a notion that models contraction and subdivision of vertices of $C$ that are labelled by a prescribed symbol 
$\epsilon\in \Sigma$.
Let $\vv$ and $\w$ be two words in $\Sigma^*$. We say that $\w$ and $\vv$ are \emph{pattern-equivalent} if one can be 
transformed one into the other by repeatedly applying one of the following operations: 
replacing the interval $\epsilon \epsilon$ by the empty word $\phi$ (\emph{even contraction}) and replacing $\phi$ by $k$
times $\epsilon\epsilon$ with $k\geq 1$ (\emph{even subdivision}). 

The pattern $\pi(\w)$ of a word $\w$ is the subword of $\w$ obtained by deleting each occurrence of $\epsilon\epsilon$ in $\w$. 

Thus the pattern of a word $\w$ is the ``shortest word'' having the same pattern as $\w$, i.e., the pattern is the subword 
of $\w$ where no further contraction is allowed. 
So, for instance, if $\w=a\epsilon\epsilon b\epsilon\epsilon\epsilon c 
\epsilon\epsilon\epsilon \epsilon b \epsilon a$, then $\pi(\w)=ab\epsilon c b \epsilon a$. 

\begin{definition}
\label{def:1}
Let $\Sigma$ be an alphabet with a distinguished special symbol $\epsilon$. 
Taking the pattern induces an equivalence relation $\approx$ on $\Sigma^*$ defined by $\uu\approx \vv$ if and only 
if $\pi(\uu)\sim \pi(\vv)$. The equivalence classes of $\Sigma^*/\approx$ are called \emph{cyclic words on $\Sigma$}
and denoted by $\cw$.
\end{definition}

As customary, the class $\cw$ containing the element $\w\in \Sigma^*$ is denoted by $\cw=[\w]$ and we say that $\w$ is a 
representative of $\cw$. It is worth noticing that for each $\w\in \Sigma^*$, $\w\approx\pi(\w)$. Since all words contained in
the same class $\cw$ have the same pattern, $\pi(\w)$ is always a representative of $\cw=[\w]$.

\comment{
\begin{remark}
Working with patterns of words on $\Sigma$ is the same as working with the elements of $\langle \Sigma \ |\ \epsilon^2\rangle$ 
i.e., the free-monoid $\Sigma^*$ with the relation $\epsilon^2=\phi$ introduced. Therefore cyclic words could have be defined as 
elements of the 
quotient $\langle \Sigma \ |\ \epsilon^2\rangle/\sim$. We actually defined cyclic words as elements of the quotient 
$\frac{\Sigma^*/\parallel}{\sim}$, where $\parallel$ is the equivalence relation on $\Sigma^*$ defined by $\uu\parallel \w$ if 
$\pi(\uu)=\pi(\vv)$. We could have defined cyclic words as elements of the quotient $\frac{\Sigma^*/\sim}{\parallel}$.
\end{remark}
}

\section{Hereditarily odd hole free multisuns}

In this section we elucidate the structure of linear matrices that are minimally unbalanced.  
\label{sec:multisun} 

\begin{lemma}
\label{lemma:1}
Let $n$ be an odd positive integer and let $\s\cong {\,\C_n\,\brack\, \B\,}$ be a linear matrix. 
Then either $g(\s)<n$ or the rows of $\B$ have either at most one or at least three nonzero entries.  
\end{lemma}
\begin{proof}
We show that if some row of $\B$ has exactly two nonzero components, then $g(\s)<n$. Let $\rb$ be such a row and observe that $\rb$ 
is not a copy of a row of $\C_n$ otherwise $\s$ would contain ${1\,1\brack 1\,1}$ as a submatrix contradicting linearity. 
We conclude that ${\,\C_n\,\brack\, \rb\,}$ is the edge-vertex adjacency matrix of a simple graph of order $n$. Such a graph, denoted by 
$C+f$, consists of a cycle $C$ of order $n$ and a chord $f$ induced by the vertices corresponding to the nonzero entries of $\rb$. 
Therefore, $n\geq 5$, because $C_3$ has no chords. Now, $C+f$ contains two cycles both containing $f$. One of these two cycles, $C'$ say, 
is odd, $C$ being odd, and shorter than $C$. Therefore, for some odd $h<n$, $\C_h$ is a submatrix of $\s$ and  $g(\s)\leq h<n$, as 
required.
\end{proof}

\begin{lemma}\label{lem:ancient}
Let $\A$ be a linear matrix of the form ${\C_n \brack \F}$. If $g=g(\A)<n$, then there exists a submatrix $\F'$  of $\A$ consisting of 
$h\geq 1$ rows of $\F$ such that ${\C_n \brack \F'}$ contains a submatrix congruent to $\C_g$ as an up-matrix.\end{lemma}
\begin{proof}
Since $g(\A)=g$, $\A$ contains an odd-cycle submatrix $\C$ of order $g$. Since $\C$ is a submatrix of $\A$ while $\C$ is not a submatrix of $\C_n$ (because $g<n$) there is a least positive integer $h$ and a matrix $\F'$ consisting of $h$ 
rows of $\F$ such that $\C$ is a submatrix of $\NN:={\C_n \brack \F'}$. Hence there is a permutation $\pi$ of 
$1,2\ldots,n,n+1,\ldots,n+h$ and a permutation $\rho$ of $1,2\ldots,n$ such that permuting the rows and the columns of $\NN$ according 
to $\pi$ and $\rho$ respectively, yields the following matrix. 
$$\TT=\begin{bmatrix}
\C & \DD\\
\E & \HH
\end{bmatrix}
$$
where $\TT\cong \NN$, and the matrices $\DD$, $\E$, and $\HH$ have appropriate dimensions. We claim that
\mybreak
\textbf{Claim.}
\emph{$[\E\ |\ \HH]$ is congruent to a submatrix $\mathsf{J}$ consisting of $n+h-g$ rows of $\C_n$. Hence the rows of $\E$ have at most 
two nonzero entries.}
\mybreak
To prove the claim it suffices to show that $\{\pi(i) \ |\ i=n+1,\ldots, n+h\}\subseteq \{1,\ldots,g\}$.
Suppose to the contrary that $\pi(n+l)>g$ for some $l$ with $1\leq l\leq h$ and let $\zeta$ be the restriction of $\pi^{-1}$ to 
$\{1,2\ldots,n,n+1,\ldots,n+h\}-\{n+l\}$. By removing the $\pi(n+l)$-th row from $\TT$ and permuting the rows and the columns of $\TT$ 
according to $\zeta$ and $\rho^{-1}$ one obtains the matrix ${\C_n \brack \F''}$ for some matrix $\F''$ consisting of $h-1$ rows of $\F$. 
Since the action of $\zeta$ and $\rho^{-1}$ on $\{1,\ldots,g\}$ sends $\C$ into one of its congruent copy, it follows that $\F''$ contains 
$\C$ as a submatrix contradicting the minimality of $h$. Hence $[\E\ |\ \HH]\cong \mathsf{J}$ for some matrix $\mathsf{J}$ consisting of 
$n+h-g$ rows of $\C_n$. Since each row of $\mathsf{J}$ has exactly two nonzero entries, it follows that each row of $\E$ has at most two 
nonzero entries.
\hfill ({\it End of Claim}) 
\mybreak
Consider now the matrix ${\C \brack \E}$ and the up-matrix $\U={\C \brack \E}^\uparrow$ of $\NN$. By the claim, the rows of $\E$ have at most two nonzero entries. But, by Lemma~\ref{lemma:1}, the non-dominated rows of $\E$ have at least three nonzero 
entries. We conclude that all rows of $\E$ are dominated. Hence $\U\cong \C$ as required.
\end{proof}

\begin{theorem}
\label{thm:mnbum}
Let $G$ be a diamond-free graph of order $n$. Then $G$ is minimally unbalanced if and only if $n\geq 5$ is odd and either $\A_G\cong \C_n$ or $\A_G\cong {\,\C_n\,\brack\, \K\,}$ for some matrix $\K$ all whose rows have at least three nonzero entries and such that 
${\C_n \brack \K'}$ does not contain as up-matrix any submatrix congruent to $\C_t$, $t<n$, for any row submatrix $\K'$ of $\K$. 
\end{theorem}
\begin{proof}
({\it If part}). Since $G$ is not balanced $g(\A_G)$ is finite, say $g(\A_G)=g$. Moreover, since $G$ is diamond-free, one has $g>3$ by Lemma~\ref{lemma:AB}. Hence $\A_G$ contains a submatrix $\C\cong \C_g$. Thus for some matrix $\F$ matrix $\DD={\,\C\,\brack\, \F\,}$ is congruent to a column submatrix of $\A_G$ . By Lemma ~\ref{lemma:1}, 
the rows of $\F$ have either at most one or at least three nonzero entries.
If all rows of $\F$ have at most one nonzero entry, then all such rows are dominated by the rows of $\C$ and  
$\DD^\uparrow\cong \C_g$. If at least one row of $\F$ has at least three nonzero entries, then $\DD^\uparrow\cong{\,\C\,\brack\, \K\,}$ where $\K$ is a row submatrix of $\F$ all whose rows have at least three nonzero entries. Moreover, in latter case, $\DD^\uparrow$ does not contain any submatrix congruent to $\C_t$ with $t<g$ and so any up-matrix congruent $\C_t$ with $t<g$, because $g(\A_G)=g$. 

Summarizing $\A_G$ contains $\DD^\uparrow$ as an up-matrix and $\DD^\uparrow$ is the clique matrix of an induced subgraph $G'$ of $G$ by Lemma~\ref{lemma:2sec}. Moreover, $G'$ is unbalanced because so is $\DD^\uparrow$. Since $G$ is minimally unbalanced it follows that $g=n$. 
\mybreak
({\it Only if part}). If $\A_G\cong \C_n$, then $n\geq 5$ by Lemma \ref{lemma:AB} and $G$ is clearly minimally unbalanced. Suppose now that for some odd $n\geq 5$, $\A_G\cong{\,\C_n\,\brack\, \K\,}$ and that ${\C_n \brack \K'}$ contains no up-matrix congruent to $\C_t$, $t<n$, for any row submatrix $\K'$ of $\K$. Since $g(\A_G)$ is finite, it follows, by Lemma~\ref{lem:ancient}, that $g(\A_G)=n$. Therefore each submatrix of $\A_G$ with less than $n$ columns is balanced. In particular so are the clique-matrices of the proper subgraphs of $G$. Therefore $G$ is minimally unbalanced.
\end{proof}

By Theorem \ref{thm:mnbum}, the essential property of a minimally unbalanced diamond-free graph is that there is a unique odd cycle 
submatrix in its clique-matrix, and such a matrix is a row submatrix. 
This fact is the easiest conclusion that one could have expected 
after the definition, because such graphs are precisely those that become balanced after removing a vertex.  

So clique-matrices of minimally unbalanced diamond-free graph are obtained as follows: start with an odd cycle matrix $\C$ and 
append 
rows to $\C$ so that the arising matrix $\s$ has the following property: $\s$ is a clique-matrix and $g(\s)$ does not decrease. 
Such a property is more easily handled and becomes more meaningful when interpreted in graphs.  
Indeed, if $G$ is a minimally unbalanced graph whose clique-matrix is of the form ${\,\C\,\brack\, \K\,}$ where $\C$ has order 
$n$ and 
$\K$ has $p$ rows, then $G$ consists of an odd cycle $C$ of order $n$ (whose edge-vertex matrix is $\C$) along with $p$ cliques 
(each one represented by a row of $\K$). The following definition is thus well justified.

\begin{definition}
\label{def:np-pair}
A \emph{multisun} is a diamond-free graph $G$ of odd order $n$ such that its maximal cliques of size 2 span a Hamiltonian cycle of $G$ called 
the \emph{rim} of $G$. All the remaining maximal cliques consist of nonconsecutive vertices of $C$ and 
are referred to as the \emph{inscribed cliques} of $G$. 
\end{definition}

In Fig.~\ref{fig:2}.a it is depicted a multisun with three inscribed cliques. 
To better understand the structure of multisuns we may use further properties of their clique-matrices. 
In fact, clique-matrices of diamond-free graphs are triangle-free (recall Lemma \ref{lemma:AB}) and linear. As a consequence, 
multisuns have linearly (in the size of the graph) many maximal cliques and such maximal cliques have the 
\emph{Helly property}, namely, any collection of pairwise intersecting maximal cliques has nonempty intersection \cite{Pr93}.
Hence, for a multisun $G$:
\begin{enumerate}[S1.]
\setlength{\itemsep}{-1mm}
\setlength{\topsep}{0pt}
\setlength{\partopsep}{0pt}
\item\label{com:nomonochromet2} The maximal cliques of $G$ are edge-disjoint and have the Helly property.
\item\label{com:cycle} The rim $C$ is uniquely determined by $G$.
\end{enumerate}
Moreover, multisuns can be recognized in polynomial time: first check for membership in the class of diamond-free graphs; 
list all the maximal cliques and check whether the maximal cliques of size 2 span a Hamiltonian cycle of $G$.

A graph is {\em odd hole free} if it does not contain any odd hole as an induced subgraph.
It turns out that odd hole freeness, though being a necessary property, is not sufficient to guarantee that a multisun is
minimally unbalanced. This because the clique-matrix of a minimally unbalanced graph has to satisfy all the requirements 
of Theorem~\ref{thm:mnbum}. For instance, in Fig.~\ref{fig:2}.b, it is shown an odd hole free multisun $G$ with a rim of 
order 23 that is not minimally unbalanced because its clique-matrix does not satisfy the conditions of Theorem~\ref{thm:mnbum}.
In fact, by deleting the edges of the rightmost inscribed triangle $T$ of $G$, one obtains a (partial) subgraph of $G$ 
consisting of an induced odd hole $C'$ of order 11. Hence, the multisun with rim $C'$ and inscribed clique $T$ is an induced 
subgraph of $G$ that is not balanced. 
To deal with these subgraphs of $G$, we introduce the following definition.

\begin{definition}
\label{def:sub-multisun}
Given a \emph{multisun} $G$, a {\em sub-multisun} of $G$ is the partial subgraph of $G$ obtained by removing the edge 
set of some (but not all) arbitrarily chosen inscribed cliques.

We say that a multisun $G$ {\em hereditarily} satisfies a given property $\mathcal P$ if it satisfies $\mathcal P$ and 
so does each of its sub-multisuns. 
\end{definition}

In what follows, we call {\em HOH-free} any multisun that hereditarily satisfies the property of being odd hole free. 
In Fig.~\ref{fig:2}.c it is depicted one of such multisuns. 
This allows us to restate Theorem~\ref{thm:mnbum} in graph-theoretical terminology as follows. 

\begin{coro}
\label{cor:2sec2}
A diamond-free graph is minimally unbalanced if and only if it is either an odd hole or an HOH-free multisun. 
\end{coro}
%

In the remaining of the paper, we provide a good characterization of HOH-free multisuns, namely, a way to describe and build the entire 
class of HOH-free multisuns and testing membership in the class efficiently. By Corollary \ref{cor:2sec2}, this will be the same as 
characterizing minimally unbalanced diamond-free graphs. 
To this aim we first elicit necessary conditions for a multisun to be HOH-free. Such conditions, referred throughout the 
rest of the paper to as {\em  N-conditions}, are listed below.

Let $C$ be the rim of a multisun $G$. Let $A$ and $B$ be two, not necessarily distinct, inscribed cliques. 
An \emph{$AB$-path} in $G$ is a subpath of $C$ whose endpoints are one in $A$ and the other in $B$ and whose 
inner vertices are in no inscribed clique. An \emph{$A$-path} in $G$ is an $AB$-path in $G$ with $A=B$; 
if $v\in B$, then an \emph{$Av$-path} is an $AB$-path whose endpoint in $B$ is $v$. Analogously, if $u\in A$, 
then a \emph{$uB$-path} is an $AB$-path whose endpoint in $A$ is $u$.
\mybreak

\paragraph{N-Conditions}
\begin{enumerate}[{N-1}]
\setlength{\itemsep}{-1mm}
\setlength{\topsep}{0pt}
\setlength{\partopsep}{0pt}
\item\label{com:p1} for each inscribed clique $A$, the number of vertices of each $A$-path is even and greater than or equal to four;
\item\label{com:p2} each inscribed clique is odd;
\item\label{com:p3} the inscribed cliques of $G$ pairwise intersect in the same vertex $\xi\in V(C)$ and are otherwise disjoint;
\item\label{com:p4} if $P$ is an $A\xi$-path for some inscribed clique $A$, then $P$ has an even number of vertices;
\item\label{com:p5} if $P$ is an $AB$-path for some two distinct inscribed clique $A$ and $B$, then $P$ has an odd number of vertices.
\end{enumerate}

\begin{figure}
    \begin{center}
             \includegraphics[width=14cm]{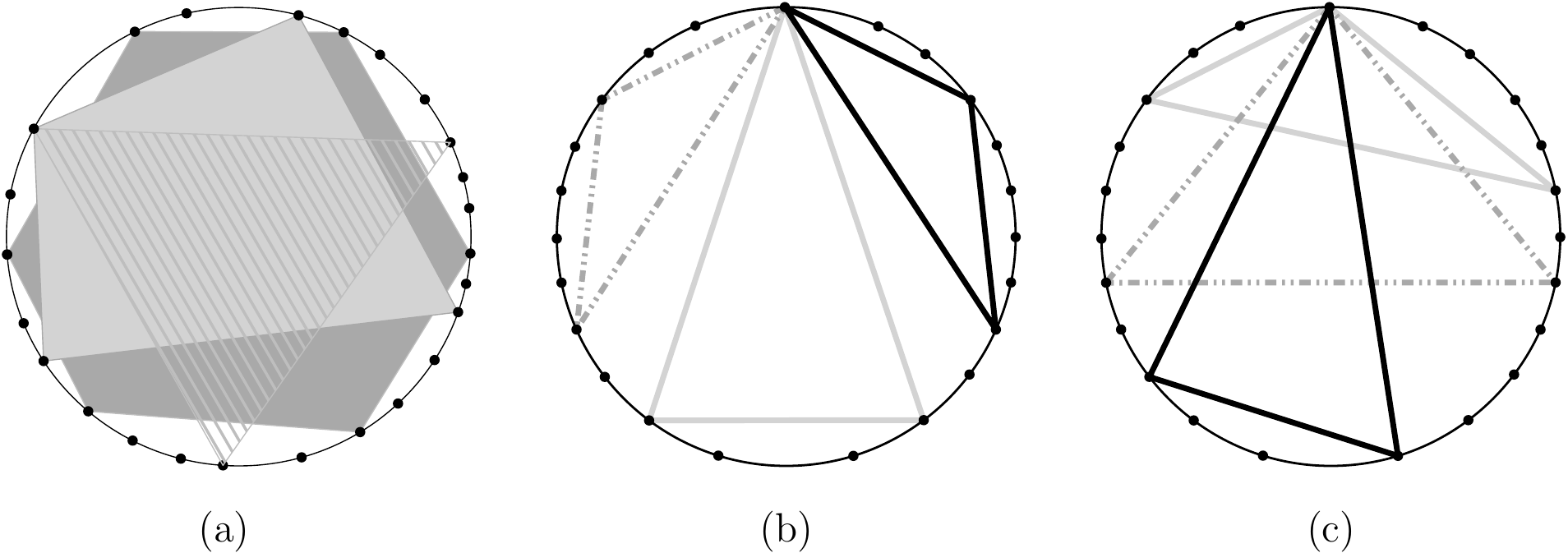}
    \end{center}
    \caption{(a) a multisun with three inscribed cliques; (b) an odd hole free multisun which is not hereditarily such; (c) a hereditarily odd-hole free multisun }
    \protect\label{fig:2}
\noindent\hrulefill%
\end{figure}

\begin{theorem}
\label{thm:fund1}
If $G$ is a HOH-free multisun, then $G$ satisfies the N-Conditions.
\end{theorem}

\begin{proof}
Let $G$ have order $n$ and $C$ be its rim. Since the only maximal cliques of $G$ are the edges of $C$ and the inscribed cliques, no 
edge of $C$ can be contained in any inscribed clique. Therefore if $u$ and $v$ are vertices that belong to the same inscribed clique 
$A$, they are not consecutive in $C$. Thus the corresponding $A$-path has at least three vertices. If it had exactly three vertices, 
then $G$ would contain a triangle using two adjacent edges of $C$ and so, both such edges would not be maximal cliques. 
Therefore the order of every $A$-path is at least four for every inscribed clique $A$. If for some inscribed clique $A$, the order 
of an $A$-path were odd, then the sub-multisun $G'$ having only $A$ as an inscribed clique, would contain an odd hole induced by 
the vertices of the $A$-path (because the endpoints of such a path are adjacent in $G'$). 
This contradicts the HOH-freeness of $G$ and establishes \eqref{com:p1}. 

To prove \eqref{com:p2}, consider the sub-multisun $G'$ of $G$ having only one inscribed clique, say $A$.
Hence $C$ is the union of $A$-paths. These $A$-paths have pairwise disjoint interiors. Therefore $V(C)-A$ in the union of these 
interiors and $n-\#A$ is even because the number of vertices of the interior of each $A$-path is such by \eqref{com:p1}. 
We conclude that $n$ and $\#A$ have the same parity. Hence each inscribed clique is odd.

To prove \eqref{com:p3} we first show that the inscribed cliques pairwise intersect. Suppose conversely that there exist two disjoint 
inscribed cliques, say $A$ and $B$. Let $G'$ be the sub-multisun of $G$ having only $A$ and $B$ as inscribed cliques. 
Label the vertices of the rim as follows: those in $A$ by $a$, those in $B$ by $b$ and those in $V(C)-A\cup B$ by $\epsilon$. 
Fix one of the two possible orientation of $C$ and choose an arbitrary vertex labeled $\epsilon$. Start traversing the cycle 
from that vertex and pause when the first vertex labeled $a$ is met. This will be the initial vertex. From this vertex traverse 
the cycle in the prescribed orientation and record the label of each vertex met during the traversal. 
Stop when the last vertex of $C$ right before the initial vertex is met. 
In this way one defines an $\{a,b,\epsilon\}$-valued sequence of the following form:
\begin{equation}\label{eq:example}
\mathbf{A}\star \mathbf{B}\star \mathbf{A} \star \mathbf{B}\cdots \mathbf{A}\star \mathbf{B} \star
\end{equation}
where $\mathbf{A}$, $\star$ and $\mathbf{B}$ are sequences defined, respectively, by
\begin{itemize}
\setlength{\itemsep}{-1mm}
\setlength{\topsep}{0pt}
\setlength{\partopsep}{0pt}
\item[--] the labels of the vertices of a maximal subpath of $C$ whose vertices are labeled either $a$ or $\epsilon$;
\item[--] the labels of interior of an $AB$-path;
\item[--] the labels of the vertices of a maximal subpath of $C$ whose vertices are labeled either $b$ or $\epsilon$.
\end{itemize}
Let $s$ be the number of occurrences of $\star$. Notice that the number of occurrences of $\star$ equals the sum of the occurrences of 
$\mathbf{A}$ and $\mathbf{B}$. Therefore, since
$\mathbf{A}$ and $\mathbf{B}$ alternate, $s$ is even. Since $A$ and $B$ are disjoint and both odd, and $n$ is
odd, the number of vertices labeled $\epsilon$ is odd. Hence there is an odd number of sequences $\star$ of odd length. But since $s$ 
is even there is at least one sequence $\star$ with even length. Hence there are two sequences $\star$ having different parity. These 
two sequences correspond to the interiors of two $AB$-paths (whose endpoints are therefore labeled $a$ and $b$). Let $I$ and $J$ be 
the vertex-sets of these two $AB$-paths. By what just said, $\#I$ and $\#J$ have different parity. We show that $I\cup J$ induces an 
odd hole in $G'$. Argue as follows: those vertices of $I\cup J$ labeled $\epsilon$ belong to neither $A$ nor $B$ while the two vertices 
labeled $a$ belong only to $A$ (therefore they are connected by an edge) and the two vertices labeled $b$ belong only to $B$ (therefore 
they are connected by an edge, as well). This contradicts that $G'$ is odd hole free and hence that $G$ is a HOH-free multisun. 
We conclude that there are no two disjoint inscribed cliques, that is, the 
inscribed cliques pairwise intersect. Since the collection of the maximal cliques of $G$ has the Helly property and 
the inscribed cliques form a subcollection consisting of pairwise intersecting members, it follows that the inscribed cliques have 
a vertex $\xi$ in common. 
On the other hand the inscribed cliques have at most one vertex in common because $G$ is diamond-free. Hence the inscribed cliques 
have exactly one vertex in common. 
We conclude that $(X-\{\xi\})\cap (Y-\{\xi\})=\emptyset$ for every two distinct inscribed cliques $X$ and $Y$ as stated. 
This completes the proof of Part \eqref{com:p3}. 

Part \eqref{com:p4} now follows by Part \eqref{com:p1}. It remains to prove Part \eqref{com:p5}. Suppose to the contrary that for 
some two distinct inscribed cliques $A$ and $B$ there is some $AB$-path $I$ with even parity. Let $u$ and $v$ be the endpoints of 
$I$ with, say, $u\in A$ and $v\in B$. By definition of $AB$-path, $u$ belongs only to $A$ and to no other inscribed clique, $v$ 
belongs only to $B$ and to no other inscribed clique while the inner vertices of $I$ belongs to no inscribed clique. It 
follows that $I\cup \{\xi\}$ induces an odd hole in $G$ because $\xi\in A\cap B$. 
This contradiction proves Part \eqref{com:p5} and therefore the theorem.
\end{proof}

\begin{remark}
Part \eqref{com:p3} of Theorem \ref{thm:fund1} was already proved in a different way in Conforti and Rao in Lemmas 5.1 and 
5.2~\cite{CR92I}.
\end{remark}

\begin{remark}
We note here explicitly that the order $n$ of a HOH-free multisun cannot be too small. Indeed $n\geq 9$ because each inscribed clique 
has at least three vertices and no triangle inscribed in a pentagon or in a heptagon can satisfy \eqref{com:p1} and 
\eqref{com:p2} with $n$ odd.
\end{remark}

The first key remark about the N-Conditions is the following. Let $G$ be a multisun and let $C$ be its rim. 
An \emph{even subdivision} of $G$ is the graph obtained by
subdividing edges of the rim through the insertion of an even number of vertices. An \emph{even contraction}
is the inverse operation of even subdivision and it is defined as follows: let $P$ be either an $A$-path or and $AB$-path of length at least 5, 
where $A$ and $B$ are two inscribed cliques; replace path $P$ by a shorter path $P'$ between the same endpoints and with the same parity of 
$P$. If $G$ is a multisun, so is each of its even contractions/subdivisions. Actually we can say more.

\begin{theorem}  If G is a multisun, then so is each of its even contractions and subdivisions.
If $G$ satisfies the N-conditions, then so does each of its even contractions and subdivisions.
\end{theorem}
\begin{proof} 
Even contractions and even subdivisions preserve each of the defining properties
of multisuns. Even contractions and even subdivisions also preserve the parity of $A$-paths,
$AB$-paths and $A\xi$-paths while they affect neither the size nor the parity of the inscribed cliques. 
\end{proof}

In view of the previous theorem we see that if $G$ is a multisun that satisfies the N-Conditions, then the class of multisuns obtained as
even contractions and/or even subdvisions of $G$ satisfies the N-Conditions as well. Therefore such properties are properties of the entire class 
rather than of the single representative. Moreover, this class always contains a ``minimal representative'', namely a multisun where
no further even subdivision/contraction is allowed. 
Such a minimal representative, called {\em standard multisun}, is a multisun whose $A$-paths and $A\xi$-paths have lenght 4 
and whose $AB$-paths have length 3 for each pair of inscribed cliques $A$ and $B$.

So, we denote by  ${\cal S}_G$ the class of all even subdivisions/contractions of a given multisun $G$. Observe that this class actually contains
all even subdivisions of the standard multisun obtained from $G$ by performing all possible even contractions.

\section{$s$-words}
\label{sec:sword}
\mybreak

Let us denote by ${\cal S}$ the class of multisuns that satisfy the N-conditions. We now aim at encoding the class of even 
contractions/subdivisions of a member of $\mathcal{S}$  by certain ``cyclic'' words (see Section~2.2).
This will allow us to handle multisuns 
symbolically and, finally, to characterize the words associated with HOH-free multisuns. 

The multisuns in ${\cal S}$  can be easily encoded by a suitable and canonical labeling of the vertices of the rim as follows. 

\begin{definition}[Canonical Labeling]
Let $G$ be a member of ${\cal S}$ and let $C$ be its rim. Let the inscribed cliques of $G$ be denoted by Latin uppercase 
letters and labeled by Latin lowercase letters $a,b,\ldots$. 
Set $\Sigma=\{\epsilon,\sigma\}\cup\{x \ |\  X\,\, \text{is an inscribed clique of $G$}\}$.
Let $f: V(C)\rightarrow \Sigma$ be defined as follows:
\begin{equation}\label{eq:flabeling}
f(v)=\left\{
\begin{array}{ll}
\sigma & \text{if}\, v\, \text { belongs to more than one inscribed clique of $G$}; \\
x & \text{if}\,\, v \,\text{ belongs to the unique inscribed clique of $G$ labeled $X$};\\
\epsilon & \text{if}\,\, v\,\text{belongs to no inscribed clique of $G$}.
\end{array}
\right.
\end{equation}
The above mapping is referred to as the \emph{canonical labeling} of $G$ and the letters in $\Sigma-\{\epsilon\}$ are called the \emph{proper letters}.
\end{definition}

Let $C$ be the rim of $G\in \mathcal{S}$. The choice of a starting vertex and of an orientation of $C$ induces a linear word on $f(V(C))$ 
(the image of the canonical labeling of $G$) that we denote by $\w_G$. Such a word is simply the sequence of the labels of the vertices met 
during the traversal 
of $C$ from the chosen starting vertex and in the prescribed direction. We can now associate with the class ${\cal S}_G$ of even subdivisions of $G$ 
a cyclic word on the alphabet $f(V(C))$ as follows. Let $G'\in {\cal S}_G$ and let $\w_G$ and $\w_{G'}$ be the ``linear'' words induced by the 
canonical labeling of $G$ and $G'$, respectively. Since $G'$ is a even subdivision/contraction of $G$, it follows that $\pi(\w_{G'})\sim \pi(\w_G)$
accordingly with Definition~\ref{def:1}. 
Hence, $\w_{G'}\approx \w_G$ and $\w_{G'}$ and $\w_G$ represent the same cyclic word $\cw_G=[\w_G]$. 

Let us examine more closely how such cyclic words $\cw_G$ look like when $G\in {\cal S}$. Suppose first that $G$ has only one inscribed clique $A$ whose 
size is odd by N-\ref{com:p2}. In this case, by suitably choosing a starting vertex,  $\w_G$ has the standard form (see Figure~\ref{fig:3}~(a))
$$a\epsilon^{\mu_1}a\epsilon^{\mu_2}a\ldots a\epsilon^{\mu_{\lambda-1}}a\epsilon^{\mu_\lambda}$$ 
where the $\mu_i$'s are all even by N-\ref{com:p4} and $\lambda$ is odd by N-\ref{com:p2}.
Therefore $\w_G\approx \pi(\w_G)= a^\lambda$ and, consequently, $\cw=[a^\lambda]$ is the cyclic word associated with $G$. 

If $G$ is a member of ${\cal S}$ that has more than one inscribed clique (see Figure~\ref{fig:3}~(b)), then, starting at $\sigma$, $\w_G$ has the following standard form:
$$\sigma \epsilon^{\mu_1}x_{i_1}^{\lambda_1} \epsilon^{\mu_2} x_{i_2}^{\lambda_2}\epsilon^{\mu_3} x_{i_3}^{\lambda_3}\cdots\epsilon^{\mu_s}x_{i_{s}}^{\lambda_{s}}\epsilon^{\mu_{s+1}},$$
where, $\mu_1$ and $\mu_{s+1}$ are both positive even integers by N-\ref{com:p4}, $\mu_i$ is a positive odd integer for $i=2,\ldots,s$ by 
N-\ref{com:p5}, the $x_{i_h}$'s are labels in $\{a,b,\ldots\}$ such that $x_{i_h}\not=x_{i_{h+1}}$ $h=1,2,\ldots,s-1$ and $\lambda_1,\ldots, 
\lambda_s$ are positive integers. Moreover, the sum of the $\lambda_i$ corresponding to each proper letter is even by 
N-\ref{com:p2}. It follows that, as in the previous case, 
$$\w_G\approx \pi(\w_G)=\sigma x_{i_1}^{\lambda_1} \epsilon x_{i_2}^{\lambda_2}\epsilon x_{i_3}^{\lambda_3}\cdots\epsilon x_{i_{s}}^{\lambda_{s}}$$
and so, $\cw_G=[\sigma x_{i_1}^{\lambda_1} \epsilon x_{i_2}^{\lambda_2}\epsilon x_{i_3}^{\lambda_3}\cdots\epsilon x_{i_{s}}^{\lambda_{s}}]$.

The above reasoning motivates the following.

\begin{figure}
    \begin{center}
             \includegraphics[width=12.5cm]{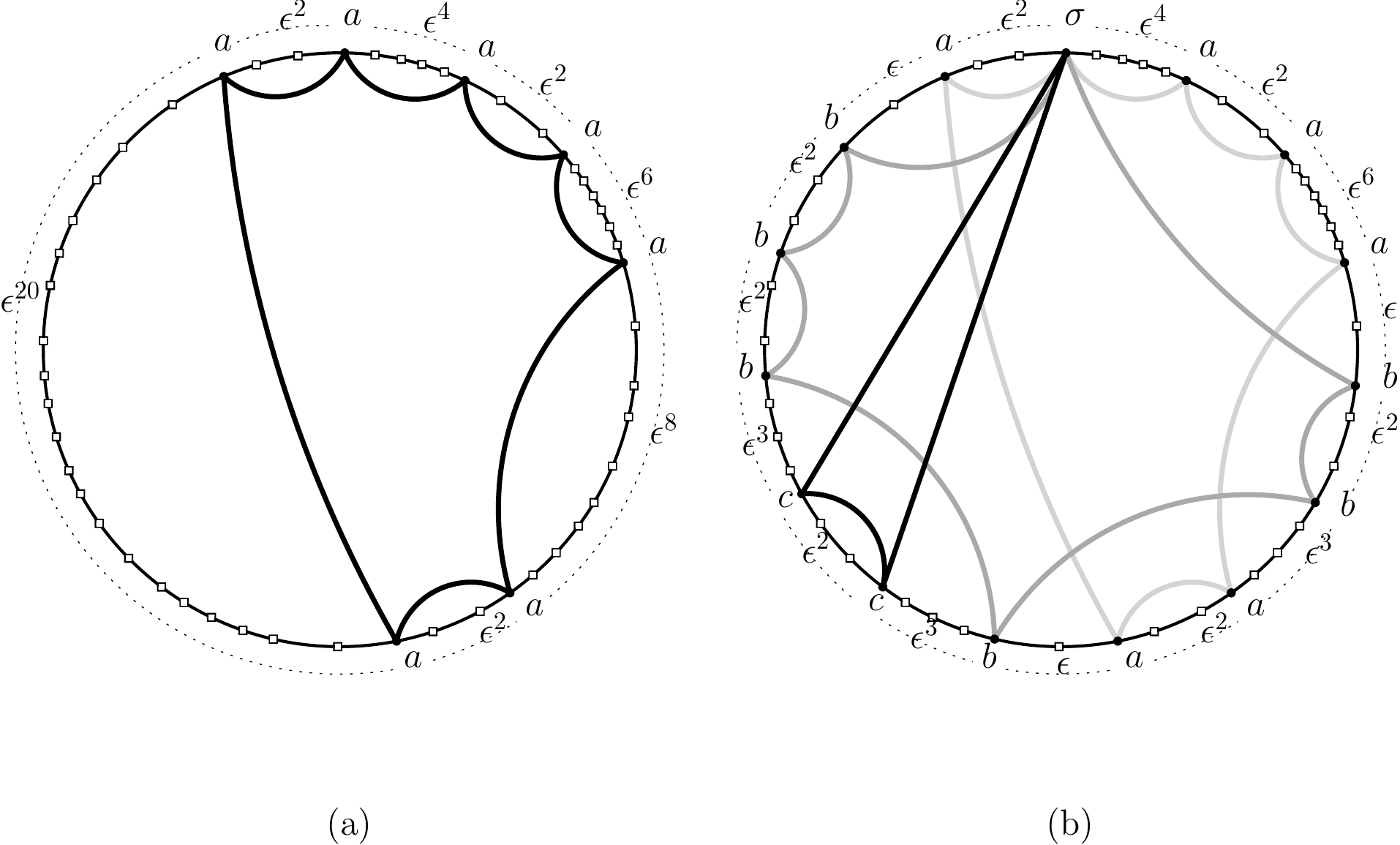}
    \end{center}
    \caption{ (a) the $s$-word $\big[a^7\big]$ is the cyclic word $[a\epsilon^2
a\epsilon^4 a\epsilon^2 a\epsilon^6 a\epsilon^8 a\epsilon^2 a\epsilon^{20}]$;
(b) the $s$-word $\big[\sigma a^3\epsilon b^2 \epsilon a^2 \epsilon b \epsilon
c^2 \epsilon b^3 \epsilon a\big]$ is the cyclic word $[\sigma \epsilon^4 a
\epsilon^2 a \epsilon^6 a \epsilon b \epsilon^2 b \epsilon^3 a \epsilon^2 a
\epsilon b \epsilon^3 c \epsilon^2 c\epsilon^3 b \epsilon^2 b \epsilon^2 b
\epsilon a \epsilon^2]$.} 
    \protect\label{fig:3}
\noindent\hrulefill%
\end{figure}

\begin{definition}[$s$-word]\label{def:swords}
\label{def:pattern} 
A cyclic word $\cw$ on the alphabet $\Sigma$ is an {\em $s$-word} if it has one of the following forms
\begin{enumerate}[{i)}]
\item\label{com:swordi} if $\Sigma=\{\epsilon,a\}$, then $\cw=[a^\lambda]$ for some odd integer $\lambda$; 
\item\label{com:swordii} if $\Sigma=\{\sigma,\epsilon,a,b,\ldots\}$, then $\cw=[\sigma x_{i_1}^{\lambda_1} \epsilon x_{i_2}^{\lambda_2}\epsilon x_{i_3}^{\lambda_3}\cdots\epsilon x_{i_{s}}^{\lambda_{s}}]$ for positive 
integers $\lambda_1,\ldots, \lambda_s$, where $x_{i_h}\not=x_{i_{h+1}}$ $h=1,2,\ldots,s-1$ and the sum of the exponents of each proper letter is even.
\end{enumerate}
\end{definition}

If $\w=\sigma x_{i_1}^{\lambda_1} \epsilon x_{i_2}^{\lambda_2}\epsilon x_{i_3}^{\lambda_3}\cdots\epsilon x_{i_{s}}^{\lambda_{s}}$, then 
$-\w=\sigma x_{i_s}^{\lambda_s} \epsilon x_{i_{s-1}}^{\lambda_{s-1}}\epsilon x_{i_{s-2}}^{\lambda_{s-2}}\cdots\epsilon x_{i_{1}}^{\lambda_{1}}$ is called the 
\emph{opposite} of $\w$. Since $-\w$ arises from $\w$ by shifting indices and reversing the order of reading, it follows that $\w\sim-\w$ 
(in particular, $\w\approx -\w$). Hence $[\w]=[-\w]$.   
\mybreak
The construction that associates members of $\mathcal{S}$ with $s$-words (via the canonical labeling) can be reversed as follows.
\mybreak
Let $\cw$ be an $s$-word and suppose first that $\Sigma$ contains at least two proper letters. 
Since $\cw$ is a cyclic word, among the representatives 
of $\cw$ there is one, say $\w$, coinciding with its own pattern and starting with $\sigma$. 
Such a representative has the form, $\sigma\uu$. Hence $\sigma\uu\approx \sigma\epsilon\epsilon\uu\epsilon\epsilon\approx 
\sigma\epsilon\epsilon\z\epsilon\epsilon$ where $\z$ is a subdivision of $\uu$ obtained by replacing any interval $xx$ of $\uu$ 
by the word $x\epsilon\epsilon x$. 
Let $\vv=v_1\ldots v_n=\sigma\epsilon\epsilon\z\epsilon\epsilon$, then the mapping $f:[n]\rightarrow \Sigma$ defined by $1\mapsto \sigma$ and 
$i\mapsto v_i$, 
$i=2,\ldots,n$ is a labeling of $C_n$. Such a labeling is the canonical labeling of a multisun $G_\cw$ with rim $C_n$ and with inscribed 
cliques $X:=\{1\}\cup f^{-1}(x)$ for $x\in \Sigma-\{\epsilon,\sigma\}$. Moreover, by construction, $\cw=[\w_{G_{\cw}}]=[\vv]$ and $G_\cw$ is the 
standard multisun of $\mathcal{S}_{G_\cw}$.
If $\Sigma=\{\epsilon,a\}$, then $\cw=[a^\lambda]$ and $G_\cw$ has a rim isomorphic to $C_{3\lambda}$ and one inscribed clique consisting of the vertices 
$1,4,7\cdots$.

\begin{definition}
For an $s$-word $\cw$, the above defined graph $G_\cw$ is called \emph{the standard multisun of} $\cw$.
\end{definition}

The straightforward construction described above proves that the class $\mathcal{S}$ and the class $\mathbb{S}$ of the $s$-words are essentialy the same 
combinatorial object. This fact is more formally summarized as follows.

\begin{prop}\label{prop:A}
$\cw$ is an $s$-word if and only if $\cw=[\w_G]$ for some $G$ in $\cal S$. If $\cw=[\w_G]$ then $G\in \mathcal{S}_{G_\cw}$.
\end{prop} 

\begin{remark}
\label{rem:A}

By Proposition \ref{prop:A}, the only thing that really matters in defining a multisun is the way in which the vertices of the inscribed cliques 
occur circularly on the rim. The essence of this fact is captured by the definition of pattern. More precisely, the class ${\cal S}_G$ is an element of the quotient of ${\cal S}$ by the relation 
\emph{$G\bowtie G'$ $\Leftrightarrow$ $G$ is an even subdivision or contraction of $G'$} and the maps 

\begin{equation}\label{eq:mapinverse}
\begin{array}{ccc}
\psi:{\cal S}/\bowtie\,\rightarrow \mathbb{S} & &\phi:\mathbb{S}\rightarrow {\cal S}/\bowtie\\ 
\qquad {\cal S}_G\mapsto [\w_G] & & \quad\cw\mapsto {\cal S}_{G_\cw}
\end{array} 
\end{equation}
are inverse of each other.
\end{remark}

Let $\cw$ be an $s$-word on an alphabet $\Sigma$ with at least two proper letters, $C$ be the rim of the standard multisun $G_\cw$ and 
$f$ the canonical labeling of $G_\cw$. For $u,\,v\in V(C)$, denote 
by $d_C(u,v)$ be the shortest path on $C$ between $u$ and $v$.
Then $\cw$ induces a linear order on $\Sigma-\{\epsilon\}$ 
denoted by $\preceq$ as follows.
\begin{subequations}
\begin{equation}\label{eq:labelorder1}
\text{$\sigma$ is the smallest element in}\,\,(\Sigma-\{\epsilon\},\preceq)
\end{equation}
and for $x\not=y$
\begin{equation}\label{eq:labelorder2}
x\preceq y \Longleftrightarrow \min_{v\in f^{-1}(x)}d_C(v,\xi)<\min_{v\in f^{-1}(y)}d_C(v,\xi)
\end{equation}
\end{subequations}
Hence we have the finite chain $\sigma\prec a\prec b\prec c\cdots$. Clearly this order depends only on the pattern of the representatives of 
$\cw$, therefore it is determined by $\cw$ and it is called the linear order {\em induced by $\cw$}. 

\begin{prop}
\label{prop:transit} 
Let $\cw$ be an $s$-word on $\Sigma$. Then the number of occurrences of $\epsilon$ in any representative $\w$ of $\cw$ is even while the number 
$s$ of the exponents of the proper letters in $\w$ is odd. Therefore such parameters are determined by $\cw$.
\end{prop}

\begin{proof}
By the definition of cyclic word, the number of occurrences of $\epsilon$ has the same parity in $\w$ 
and $\pi(\w)$ while the number of exponents of the proper letters is the same in $\w$ and $\pi(\w)$. 
If $\cw=[a^\lambda]$ there is nothing to prove. We can thus suppose that $\cw$ is as in \ref{com:swordii}) of 
Definition~\ref{def:swords}. 
Observe that $1+\sum_{i=1}^s\lambda_i+(\text{number of occurrences of $\epsilon$}$)
has the same parity as the order of the rim of $G_\cw$ and such an order is odd by the definition of multisun (the term 1 accounts for $\sigma$). 
Since $\sum_{i=1}^s\lambda_i$ is even by N-\ref{com:p2}, it follows that the number of occurrences of $\epsilon$ is such while $s$ is odd.
\end{proof}

Our last device is a formal way to encode the operation of taking sub-multisuns. Recall that, if $G$ is a multisun with $p$ inscribed cliques, 
such an operation consists of removing the edge-set of a subset of $q<p$ cliques (so we are not allowed to remove all the inscribed cliques). 

Let thus $\cw$ be an $s$-word on $\Sigma$ and let $\w$ be any of its representatives. Assume that $\Sigma$ contains $p$ proper letters and 
$p\geq 2$. Hence $\sigma\in \Sigma$. If $\{x_1,\ldots,x_q\}\subseteq \Sigma-\{\epsilon,\sigma\}$, then we denote by
$\w\arrowvert_{x_1=\epsilon,\ldots x_q=\epsilon}$ the linear word $\uu$ on $(\Sigma-\{x_1,\ldots,x_q\})$ defined as follows:
\mybreak
--if $q<p-1$, then $\uu$ is obtained from $\w$ by setting to $\epsilon$ all the occurrences of the letters in the set $\{x_1,\ldots,x_q\}$ 
\mybreakk
--if $q=p-1$, then $\uu$ is obtained from $\w$ by setting to $\epsilon$ all the occurrences of the letters in the set $\{x_1,\ldots,x_q\}$ and by setting $\sigma$ to $y$ where $y$ is the unique proper letter in $(\Sigma-\{x_1,\ldots,x_q\})$.
\mybreak
For instance, if $\w=\sigma a\epsilon b \epsilon c^2\epsilon b \epsilon a^3$ then $\w\arrowvert_{a=\epsilon}=\sigma\epsilon^2 b \epsilon c^2\epsilon b \epsilon^4$ while $\w\arrowvert_{a=\epsilon, b=\epsilon}=c\epsilon^4 c^2\epsilon^6$.

We now set 
$$\cw\arrowvert_{x_1=\epsilon,\ldots x_q=\epsilon}=[\w\arrowvert_{x_1=\epsilon,\ldots x_q=\epsilon}]$$

and call it the \emph{projection} of $\cw$ on $\Sigma -\{x_1,\ldots,x_q\}$. It is easily checked that 
$$[\w]=[\vv]\Longrightarrow [\w\arrowvert_{x_1=\epsilon,\ldots x_q=\epsilon}]=[\vv\arrowvert_{x_1=\epsilon,\ldots x_q=\epsilon}].$$ 
Therefore projection is a well defined operation. Referring to the examples above we have, $\cw\arrowvert_{a=\epsilon}=[\sigma b \epsilon c^2\epsilon b]$ (because $\sigma\epsilon^2 b \epsilon c^2\epsilon b \epsilon^4$ $\approx$ $\pi(\sigma\epsilon^2 b \epsilon c^2\epsilon b \epsilon^4)$)and $\cw\arrowvert_{a=\epsilon, b=\epsilon}=[c^3]$ (because $c\epsilon^4 c^2\epsilon^6$ $\approx$ $\pi(c\epsilon^4 c^2\epsilon^6)=c^3$).
Notice that if $\Sigma$ contains exactly one proper letter, then the projection is not defined: in this case we would have $q=p$. 
\mybreak

\section{Sunoids and Sunwords}
\label{sec:sunword}

In this section, we focus on a subclass of ${\cal S}$ that will turn out to be equivalent to HOH-free multisuns:
the class of multisuns in ${\cal S}$ that hereditarily satisfy the N-conditions. 

\begin{definition}[Sunoids]
\label{def:fsunoid}
A \emph{sunoid} is a multisun $G$ that hereditarily satisfies the N-conditions, i.e., each sub-multisun of $G$ satisfies the N-conditions. 
The class of sunoids is denoted by ${\cal S}^*$. 
\end{definition}

If ${\mathcal S}^{**}$ denotes the class of HOH-free multisuns, then we clearly have the containment among classes 
${\mathcal S}^{**}\subseteqq {\mathcal S}^*\subseteqq {\mathcal S}$. 
In the next section we actually prove that ${\mathcal S}^{**}={\mathcal S}^*$.

\comment{
By \eqref{eq:mapinverse}, the $s$-words are precisely the images of the classes $\mathca{S}_G$ within cyclic words. 
When $G$ runs within sunoids, the corresponding words describe a subclass of $s$-words. If one recalls the definition of canonical labeling, 
then it is clear that setting letter $x$ to $\epsilon$ in any representative of an $s$-word of a sunoid amounts to deleting the edge set of 
the inscribed clique whose label is $x$. Hence the $s$-word of a sunoid has the property that, as long as the word contains at least two 
proper letters, by setting to $\epsilon$ all the occurrences of any given proper letter results in an $s$-word (projection). 
}
As observed in Section~\ref{sec:sword},  projecting is the same as taking sub-multisuns (up to some technicalities).
Since sunoids hereditarily satisfy the N-conditions their corresponding $s$-words are hereditarily $s$-words.
In view of this property, $s$-words of sunoids deserve a special name.

\begin{definition}
A sunword is the $s$-word of a sunoid. The class of sunwords is denoted by $\mathbb{S}^*$
\end{definition}

Both the $s$-words in Fig.~\ref{fig:3} are examples of sunwords.
Hence, characterizing sunoids among multisuns satisfying the N-conditions is the same as characterizing sunwords among $s$-words. 
This will be done in the following where we show that sunwords are the $s$-words that satisfy certain parity conditions on the exponents 
of the proper letters plus a sort of ``continuity'' property with respect to a linear order on the proper letters. 
By Proposition \ref{prop:A} this yields a characterization of sunoids via the maps in \eqref{eq:mapinverse}. 

The next property is the essential property of sunwords within the $s$-words. To understand it, it suffices to observe that taking a 
sub-multisun is tantamount that projecting the corresponding sunword on the complement of the labels of the cliques removed and resort to the fact 
that sunoids are hereditarily such. 

\begin{prop}\label{prop:B}
Every projection of a sunword is a sunword. The class $\mathbb{S}^*$ is thus closed under projection.
\end{prop}

\begin{proof}
Let $\cw$ be an $s$-word with at least two proper letters. Let $F_X$ be the edge-set of the inscribed clique $X$ of the standard multisun of 
$\cw$ and let $x$ be the label of $X$. Moreover, let $\cw'=\cw\arrowvert_{x=\epsilon}$ and let $G'$ be the sub-multisun obtained form $G_\cw$ by 
deleting the edges of $F_X$ (remark that $G'$ is not a standard multisun). Finally let $\psi$ and $\phi$ be the applications defined by 
\eqref{eq:mapinverse}. The following diagram proves now the statement
$$\xymatrix@1{
\cw  \ar@{|->}[rr]^{\phi} && {\cal S}_{G_\cw}\ar@{|->}[rr]^{G_\cw-F_X} && {\cal S}_{G'}\ar@{|->}[rr]^{\psi} && \cw'}.$$
Since $\psi$ takes the class ${\cal S}_{G'}$ of any multisun to the corresponding $s$-word, it takes the class of a sunoid into its sunword. 
Therefore $\cw'$ is a sunword, because $G'$ is a sunoid. 
\end{proof}
We are now in position to begin with the characterization of sunwords.
\begin{lemma}
\label{lemma:00}
Let $\cw$ be a sunword and $\sigma\prec a\prec b\cdots $ be the induced linear order on $\Sigma-\{\epsilon\}$. Then $\cw$ admits a representative 
$\w=\sigma a^{\lambda_1} \epsilon x_{i_2}^{\lambda_2}\epsilon x_{i_3}^{\lambda_3}\cdots\epsilon x_{i_{s}}^{\lambda_{s}}$ where $x_{i_2}\not=a$ and such that $\w$ satisfies the following \emph{parity conditions}: 
\begin{enumerate}[{\rm (i)}]
\item\label{com:001} $\lambda_1$ and $\lambda_s$ are both odd; 
\item\label{com:002} $\lambda_h$ is odd if and only if $x_{i_{h-1}}\ne x_{i_{h+1}}$, $h=2,\ldots,s-1$;
\end{enumerate}
\end{lemma}
\begin{proof}
It is clear that a representative $\w$ with $\sigma$ followed by $a$ exists: $a$ is the label of the vertex $v^*$ of $G_\cw$ closest to $\sigma$ 
among those that are not labeled $\epsilon$; thus, by walking along the rim of $G_\cw$ starting from $\sigma$ toward $v^*$, one gets a word whose 
pattern represents $\cw$ and where $a$ follows $\sigma$. Let us prove that $\cw$ satisfies the parity conditions. Since $\cw$ is hereditary, 
the projection $\cw'$ of $\cw$ on $\Sigma-\{a\}$ is a sunword. Hence 
$$\cw'=[\sigma \epsilon^{\lambda_1} \epsilon x_{i_2}^{\lambda_2}\epsilon x_{i_3}^{\lambda_3}\cdots\epsilon x_{i_{s}}^{\lambda_{s}}]=
[\sigma \epsilon^{\lambda_1+1} x_{i_2}^{\lambda_2}\epsilon x_{i_3}^{\lambda_3}\cdots\epsilon x_{i_{s}}^{\lambda_{s}}]$$ 
and $\sigma\epsilon^{\lambda_1+1}x_{i_2}^{\lambda_2}$ is the image of $\xi X_{i_2}$-path in $G_{\cw'}$. Therefore
by N-\ref{com:p4}, $\lambda_1+1$ has to be even and $\lambda_1$ is odd. Since $-\w=
\sigma x_{i_{s}}^{\lambda_{s}}\epsilon\cdots x_{i_3}^{\lambda_3}\epsilon x_{i_2}^{\lambda_2} \epsilon a^{\lambda_1}$ (the opposite of $\w$) represents the 
same cyclic word as $\w$, it follows that by applying the same reasoning to the leftmost word in the chain above, one concludes that 
$\lambda_{s+1}$ is odd and this establishes (i).
\mybreak
To prove (ii), observe that the interval $x_{h-1}^{\lambda_{h-1}}\epsilon x_h^{\lambda_{h}}\epsilon x_{h+1}^{\lambda_{h+1}}$ of $\cw$ maps to the interval 
$x_{h-1}^{\lambda_{h-1}}\epsilon^{1+\lambda_{h}+1}x_{h+1}^{\lambda_{h+1}}$ of  
$\cw'=[\w\arrowvert _{x_{h}=\epsilon}]$. Since $\cw'$ is a sunword, then such an interval is the image of an $X_{h-1}X_{h+1}$-path in the sunoid 
$G_{\cw'}$. Therefore $\lambda_h+2$ is odd if and only if $x_{h-1}\ne x_{h+1}$ by N-\ref{com:p5}. But $\lambda_h+2$ has the same parity as 
$\lambda_h$ and (ii) follows. 
\end{proof}

\begin{lemma}
\label{lem:2lettere}
Let $\cw$ be a sunword on the alphabet $\Sigma$ and let $\sigma\prec a\prec b\prec c\cdots $ be the induced linear order on $\Sigma-\{\epsilon\}$.
\begin{itemize} 
\item[--] If $\Sigma=\{\epsilon,a\}$, then $\cw=[a^\lambda]$ for some positive odd integer $\lambda$. 
\item[--] If $\Sigma=\{\epsilon,\sigma,a,b\}$, then 
$$\cw=[\sigma a^{\lambda_1}\epsilon b^{\lambda_2}\epsilon a^{\lambda_3} \cdots b^{\lambda_{s-1}}\epsilon a^{\lambda_{s}}]$$
where $\lambda_1$ and $\lambda_s$ are both odd while all the other $\lambda_i$'s are even.
\item[--] If $\Sigma=\{\epsilon,\sigma,a,b,c\ldots\}$, then
\begin{equation}\label{eq:wordform}
\cw=[\sigma a^{\lambda_{1}}\epsilon\widetilde{\w}\epsilon a^{\lambda_{s}}]
\end{equation}
where $\lambda_1$ and $\lambda_s$ are both odd and $\widetilde{\w}\in (\Sigma-\{\sigma\})^\ast$ is of the form 
$b^{\lambda_{2}}\epsilon \uu \epsilon b^{\lambda_{s-1}}$. 
\end{itemize}
\end{lemma}

\begin{proof}
If $\Sigma=\{\epsilon,a\}$, then each $s$-word is a sunword because no projection is possible. Therefore $\cw$ has the stated form by 
Definition~\ref{def:swords}. 
\mybreak
Suppose now that $\Sigma=\{\epsilon,\sigma,a,b\}$ with $\sigma\prec a\prec b$. 
Let $\w$ be a representative of $\cw$ as in Lemma \ref{lemma:00}. 
Accordingly, $\w$ is either of the following forms:
\begin{subequations}
\begin{equation}\label{eq:w1}
\sigma a^{\lambda_1}\epsilon b^{\lambda_2}\epsilon a^{\lambda_3}\cdots a^{\lambda_{s-1}}\epsilon b^{\lambda_{s}},
\end{equation}
\begin{equation}\label{eq:w2}
\sigma a^{\lambda_1}\epsilon b^{\lambda_2}\epsilon a^{\lambda_3} \cdots b^{\lambda_{s-1}}\epsilon a^{\lambda_{s}}
\end{equation}
\end{subequations} 
The number $s$ of exponents of the proper letters is odd by Proposition \ref{prop:transit}. Hence $\w$ cannot have the form \eqref{eq:w1} 
because $a$ and $b$ occur the same number of times implying that $s$ is even. Therefore $\w$ is of the form in \eqref{eq:w2} as stated. 
The fact that $\lambda_1$ and $\lambda_s$ are both odd while all the other exponents are even now follows directly from Lemma \ref{lemma:00}.
\mybreak
Suppose finally that $\Sigma=\{\epsilon,\sigma,a,b,c,\cdots\}$ with $\sigma\prec a\prec b\prec c\cdots$. 
As above, let $\w$ be a representative of $\cw$ as in Lemma \ref{lemma:00}.
Let us prove that $\w=\sigma a^{\lambda_1}\epsilon\widetilde{\w}\epsilon a^{\lambda_s}$ where $\widetilde{\w}\lhd \w$. 

We argue as follows. If $\w=\sigma a^{\lambda_1}\epsilon \w'\epsilon x^{\lambda_s}$, for some $x\in \Sigma-\{\epsilon,\sigma,a\}$, and $\w'\lhd \w$, 
then the projection of $\w$ on $\Sigma - \{\epsilon,\sigma,a,x\}$ would not be of the form \eqref{eq:w2}. 
Hence $\w=\sigma a^{\lambda_1}\epsilon x_{i_2}^{\lambda_2}\epsilon x_{i_3}^{\lambda_3}\epsilon \cdots\epsilon x_{i_{s-1}}^{\lambda_{s-1}}\epsilon a^{\lambda_s}$ 
with $x_{i_2}\not=a$, $x_{i_{s-1}}\not=a$ and $\widetilde{\w}=x_{i_2}^{\lambda_2}\epsilon x_{i_3}^{\lambda_3}\epsilon \cdots\epsilon x_{i_{s-1}}^{\lambda_{s-1}}$. 
This establishes \eqref{eq:wordform}. 

It remains to show that $x_{i_2}=x_{i_{s-1}}=b$. Set $\w'=\widetilde{\w}\arrowvert_{a=\epsilon}$. By projecting onto $\Sigma-\{a\}$, after recalling 
that $\lambda_1$ and $\lambda_s$ are both odd, one gets

$$\w\arrowvert_{a=\epsilon}=\big(\sigma a^{\lambda_1}\epsilon x_{i_2}^{\lambda_2}\epsilon x_{i_3}^{\lambda_3}\epsilon \cdots\epsilon x_{i_{s-1}}^{\lambda_{s-1}}\epsilon a^{\lambda_s}\big)\arrowvert_{a=\epsilon}\approx \sigma \widetilde{\w}\arrowvert_{a=\epsilon}\approx \sigma x_{i_2}^{\rho_1}\epsilon\ldots \epsilon x_{i_s}^{\rho_r}=\sigma\w'$$

for some integer $r$ with $r\leq s-2$ and some integers $\rho_1,\ldots,\rho_r$ such that $\rho_1\geq \lambda_2$ and $\rho_r\geq \lambda_{s-1}$. 
This because setting $a$ to $\epsilon$ shorten the pattern of $\widetilde{\w}$ and might affect the exponents of the proper letters. For instance, 
if $a$ is interlaced by the same proper letter $z$, then once $a$ is set to $\epsilon$ the exponent of $z$ might increase. 
Now $\cv=[\sigma\w']$ is a sunword being the projection of a sunword. Therefore, by \eqref{eq:wordform}, 
$\w'= b^{\rho_1}\epsilon \uu \epsilon b^{\rho_r}$ where $\uu\in (\Sigma-\{\sigma,a\})^\ast$ and the lemma follows. 
\end{proof}

Sunwords exhibit a strong symmetrical shape and the intervals like $a^{\lambda_2}\epsilon b^{\lambda_3}\epsilon\uu \epsilon b^{\lambda_{s-1}}\epsilon a^{\lambda_1}$ appear moderately palindrome. What is however remarkable, is that the linear order $\prec$ is made compatible with the 
action of the dihedral group on the rim, by forcing the vertices labeled by the proper letters of $\Sigma$ to be ranked by the distance 
from $\xi$ regardless of the orientation we choose. That is, if $\w$ is a representative of a sunword as in Lemma~\ref{lemma:00}, then if a proper letter 
$y$ occurs for the first time after $x$, then $y$ occurs for the first time after $x$ in $-\w$ as well. 
This can be deduced by repeatedly projecting on a set of two 
proper letters. However, sunwords have an even 
stronger structure, namely if $y$ occurs for the first time right after $x$ and $z$ occurs right after $y$ then $y$ cannot appear between 
letters different from $x$ and $z$, that is, sunwords behave with a sort of ``continuity''. 
This crucial property is formalized in the following. 


\begin{definition}
Let $\cw=[x_{i_0} x_{i_1}^{\lambda_1} \epsilon x_{i_2}^{\lambda_2}\epsilon x_{i_3}^{\lambda_3}\cdots\epsilon x_{i_{s}}^{\lambda_{s}}]$ be an $s$-word 
where $x_{i_0}=\sigma$. Let $\sigma\prec a\prec b\cdots $ be the linear order induced  by $\cw$ on $\Sigma-\{\epsilon\}$. Two letters $x$ and $y$ of $\Sigma-\{\epsilon\}$ form a \emph{cover} pair if $x\preceq y$ and there does not exist a proper letter $z\ne x,y$ such that $x\prec z\prec y$.

Two letters $x_{i_h}$ and $x_{i_{h+1}}$, $h\geq 0$ (sums are modulo $s+1$) are a \emph{jump on $x$ and $y$} in $\cw$ if $x$ 
and $y$ is not a cover pair, and either $x_{i_h}=x$, $x_{i_{h+1}}=y$ or $x_{i_h}=y$, $x_{i_{h+1}}=x$ . An $s$-word is \emph{jump-free} if it contains no jump for any $\{x,\,y\}\subseteq\Sigma-\{\epsilon\}$.
\end{definition}

We remark here explicitly that if $\preceq$ is the order induced by $\cw$ and if $\cw'$ is a projection of $\cw$ onto $\Sigma'\subseteq \Sigma$, 
then the linear order $\preceq'$ induced by $\cw'$ on $\Sigma'-\{\epsilon\}$ is precisely the restriction of $\preceq$ on $\Sigma'$. 
Therefore, with some abuse of language, we use the same symbol for the order $\preceq$ and its restrictions, because this does not cause 
confusion.

\begin{theorem}
\label{thm:fund2}
If $\cw$ is a sunword, then $\cw$ is jump-free. 
\end{theorem}
\begin{proof}
The statement is clearly true when $\cw$ has at most two proper letters because of the second part of Lemma~\ref{lem:2lettere}.  

We therefore assume that the alphabet $\Sigma$ has at least three proper labels. Moreover, since if $\cw$ contains a jump on $x$ and $z$, then the 
projection of $\cw$ on $\Sigma -\{\epsilon,\sigma,x,y,z\}$ contains a jump on $x$ and $z$ for each $y$ such that $x\prec y\prec z$, 
it is not restrictive to prove the theorem when $\Sigma$ is $\{\epsilon,\sigma,a,b,c\}$ with $a\prec b\prec c$. 
So, suppose by contradiction that $\cw$ has a jump. Such a jump cannot be on $\sigma$ and $x$ for $x\in \{b,c\}$ because $\sigma$ occurs only once 
and because of the third part of Lemma~\ref{lem:2lettere}. Therefore it must be jump on $a$ and $c$. Hence, if 
$\w$ is a representative of $\cw$ defined as in \eqref{eq:wordform}, then such a jump occurs on the right of $\sigma$. 
It follows that there exists an integer $t$ such that the prefix
$$\uu_0=\sigma a^{\lambda_1}\epsilon\cdots \epsilon x_{i_t}^{\lambda_{t}}$$
is jump-free while
$$\uu=\sigma a^{\lambda_1}\epsilon\cdots \epsilon x_{i_t}^{\lambda_{t}}\epsilon x_{i_{t+1}}^{\lambda_{t+1}}$$
is not jump-free.
Therefore, either $x_{i_t}=a$, $x_{i_{t+1}}=c$ and $\uu=\sigma a^{\lambda_1}\epsilon\cdots \epsilon a^{\lambda_{t}}\epsilon c^{\lambda_{t+1}}$ or $x_{i_t}=c$, $x_{i_{t+1}}=a$ and $\uu=\sigma a^{\lambda_1}\epsilon\cdots \epsilon c^{\lambda_{t}}\epsilon a^{\lambda_{t+1}}$.
\mybreak
{\bf  Case 1. $\uu=\sigma a^{\lambda_1}\epsilon\cdots \epsilon a^{\lambda_{t}}\epsilon c^{\lambda_{t+1}}$}. 

\noindent Since $a\prec b\prec c$, $t\geq 3$ by Lemma \ref{lem:2lettere}.
If $\lambda_t$ is even, then Lemma \ref{lemma:00} implies that $\uu$ contains the postfix 
$c^{\lambda_{t-1}}\epsilon a^{\lambda_{t}}\epsilon c^{\lambda_{t+1}}$, contradicting the choice of $t$.  
Hence, $\lambda_t$ is odd and $\uu$ contains the postfix $b^{\lambda_{t-1}}\epsilon a^{\lambda_{t}}\epsilon c^{\lambda_{t+1}}$. 
Let $\q=\vv\epsilon b^{\lambda_{t-1}}\epsilon a^{\lambda_{t}}\epsilon c^{\lambda_{t+1}}$
be the longest postfix of $\uu$ with the property that $\vv$ is an interval on $\{\epsilon,a,b\}$. 

If $\vv$ starts with $a$, then, by the maximality of $\q$, it is not difficult to see that $\uu$ is of the following form
\begin{equation}\label{eq:qui2}
\uu= \sigma \epsilon \dots \underbrace{a^{\lambda_1}\epsilon b^{\lambda_2}\cdots}_{\vv}\epsilon b^{\lambda_{t-1}}\epsilon a^{\lambda_{t}}\epsilon c^{\lambda_{t+1}}.
\end{equation}

Thus $\uu\arrowvert_{b=\epsilon}=\sigma\epsilon a^{\lambda'}\epsilon c^{\lambda_{t+1}}$,
where $\lambda'$ is the sum of the exponents $\lambda_i$, $i\geq 1$, of $a$ in \eqref{eq:qui2}. Among such exponents $\lambda_1$ and $\lambda_t$ are both odd: the former by \eqref{com:001} of Lemma~\ref{lemma:00}, while the latter by the assumption. All the other exponents are even 
by  \eqref{com:002} of Lemma~\ref{lemma:00}. It follows that $\lambda'$ is even. 
Since $\uu\arrowvert_{b=\epsilon}\lhd \w\arrowvert_{b=\epsilon}$,  
$\lambda'$ even implies that $\w\arrowvert_{b=\epsilon}$ contradicts ~\eqref{com:001} of Lemma~\ref{lemma:00}.  

It follows that $\vv$ starts with $b$. Hence, there is an integer $l> 3$ such that 

$$c^{\lambda_l}\epsilon\q=c^{\lambda_l}\epsilon\vv\epsilon b^{\lambda_{t-1}}\epsilon a^{\lambda_t}\epsilon c^{\lambda_{t+1}}\lhd\uu\lhd \w.$$

Since the interval $\vv\epsilon b^{\lambda_{t-1}}$ in the above formula is a word on $\{a,b\}$ starting and ending with $b$,  
the exponents of all the occurrences of $a$ in $\vv$ are even. 
But now the exponent $\lambda''$ of $a$ in 

$$c^{\lambda_l}\epsilon \q\arrowvert_{b=\epsilon}=c^{\lambda_l}\epsilon  a^{\lambda''}\epsilon c^{\lambda_{t+1}}$$

is the sum of $\lambda_t$ plus the exponents of $a$ in $\vv$. Since the latter are all even, $\lambda''$ has the same parity of $\lambda_t$ and therefore it is odd, $\lambda_t$ being odd. But this again contradicts \eqref{com:002} of Lemma~\ref{lemma:00}. 
\mybreak
{\bf Case 2. $\uu=\sigma a^{\lambda_1}\epsilon\cdots \epsilon c^{\lambda_{t}}\epsilon a^{\lambda_{t+1}}.$}

Let $$\q=\vv\epsilon c^{\lambda_t}\epsilon a^{\lambda_{t+1}}$$ 
be the longest postfix of $\uu$ with the property that $\vv$ is an interval on $\{b,c\}$. The minimality of $t$ and the maximality of $\q$ (and hence of $\vv$) imply that $\vv$ starts and ends with $b$. Hence $\lambda_t$ is odd and, for some $r\geq 3$, 
\begin{equation}\label{eq:case}
\z:=a^{\lambda_r}\epsilon \vv\epsilon c^{\lambda_t}\epsilon a^{\lambda_{t+1}}\lhd \uu\lhd \w.
\end{equation}
Now the exponent $\theta$ of $c$ in $\z\arrowvert_{b=\epsilon}$ is the sum of $\lambda_t$ and the exponents of $c$ in $\vv$ which are all even by 
\eqref{com:002} of Lemma~\ref{lemma:00}. Hence $a^{\lambda_r}\epsilon  c^\theta\epsilon a^{\lambda_t}\lhd \w\arrowvert_{b=\epsilon}$ with $\theta$ odd and this contradicts \eqref{com:002} of Lemma~\ref{lemma:00}. 

We conclude that $\cw$ is is jump-free.
\end{proof}

\begin{coro}\label{coro:parallel}
Let $\cw$ be a sunword and let $\w$ be a representative of $\cw$ as in \eqref{eq:wordform} of Lemma~\ref{lem:2lettere}. Let $z$ be the greatest 
element in $(\Sigma-\{\epsilon\},\prec)$. Then the exponents of $z$ in $\w$ are all even.
\end{coro}
\begin{proof}
Since a sunword is jump-free, it follows that the letters before and after each power of $z$ in $\w$, coincide. The result now follows 
from Lemma~\ref{lemma:00}.  
\end{proof}

\begin{coro}\label{coro:parallel1}
Let $\cw$ be a sunword and let $\uu$ be an interval of any representative $\w$ of $\cw$. Let $l(\uu)$ and $m(\uu)$ be the lowest and greatest 
elements of $(\spt(\uu)-\{\epsilon\},\prec)$. Then $\spt(\uu)$ contains all the letters of $\Sigma-\{\epsilon\}$ between $l(\uu)$ and $m(\uu)$.
\end{coro}
\begin{proof}
Just observe that if a letter is missing then $\uu$ contains a jump.
\end{proof}

We have just proved that if $\cw$ is a sunword, then $\cw$ satisfies the parity conditions (Lemma~\ref{lemma:00}) and is 
jump-free (Theorem~\ref{thm:fund2}), i.e., these conditions are necessary for an $s$-word to be a sunword. 
We now show that such conditions are also sufficient and therefore characterize sunwords within $s$-words (Theorem~\ref{thm:fund3}). 
To this end we need some more intermediate results.

\begin{lemma}
\label{lem:suff1}
If $\cw$ is an $s$-word on at least two proper letters and $\cw$ satisfies the parity conditions, then 
$\cw\arrowvert_{y=\epsilon}$ is an $s$-word for each proper letter $y$.
\end{lemma}

\begin{proof}
Let $\cw=[\w]$ where $\w=\sigma x_{i_1}^{\lambda_1} \epsilon x_{i_2}^{\lambda_2}\epsilon x_{i_3}^{\lambda_3}\cdots\epsilon x_{i_{s}}^{\lambda_{s}}$. 
Clearly, $y$ occurs among the $x_{i_j}$'s and each time it occurs, it is located between two (not necessarily distinct) letters of
$\Sigma-\{\epsilon\}$. 
\mybreak
Suppose first that $x_{i_1}\not=y$ and $x_{i_s}\not=y$. In this case, $\w$ can be written as follows

\begin{equation}\label{eq:nuova1}
\sigma\x_1\epsilon\x_2\epsilon\x_3\epsilon\cdots\epsilon \x_r
\end{equation}

where $r\leq s$ and for $j=1,\ldots r$, $\x_j$ is an interval of $\w$ such that $\spt(\x_j)\subseteq \{\epsilon,x,y\}$ for some 
proper letter $x\not=y$ and $\x_j$ is defined as follows:
\begin{itemize}
\setlength{\itemsep}{-1mm}
\setlength{\topsep}{0pt}
\setlength{\partopsep}{0pt}
\item[--] 
\label{com:type1}  $\x_j=y^{\lambda_m}$ if $x_{i_{m-1}}$ and $x_{i_{m+1}}$ are
two distinct proper letters both different from $y$ for some $m\geq j$;
\item [--] 
\label{com:type2} $\x_j=x^{\lambda_{l}}\epsilon y^{\lambda_{l+1}}\epsilon x^{\lambda_{l+2}}\epsilon y^{\lambda_{l+3}}\cdots \epsilon x^{\lambda_{m}}$ 
if both  $x_{i_{l-2}}$ and $x_{i_{m+2}}$ are different from $x$, for some $l$ and $m$ such that $j\leq l\leq m$. Notice that if 
$y\not\in\spt(\x_j)$ then $l=m$.
\end{itemize}

The definition of $\x_j$ and the parity conditions imply that
\begin{enumerate}[a)]
\setlength{\itemsep}{-1mm}
\setlength{\topsep}{0pt}
\setlength{\partopsep}{0pt}
\item\label{com:n0} if $\x_j=y^{\lambda_m}$ then $\lambda_m$ is odd;
\item\label{com:void} if $\x_j=x^{\lambda_{l}}\epsilon y^{\lambda_{l+1}}\epsilon x^{\lambda_{l+2}}\epsilon y^{\lambda_{l+3}}\cdots \epsilon x^{\lambda_{m}}$, then
\begin{enumerate}[1.]
\item\label{com:ni} $\lambda_h$ is even for $h=l+1\ldots,m-1$;
\item\label{com:nii} if $x_{i_{l-1}}\ne y$ then $\lambda_{l}$ is odd;
 analogously if $x_{i_{m+1}}\ne y$ then $\lambda_{m}$ is odd.
\end{enumerate}
\end{enumerate}
Let $\w'=\w\arrowvert_{y=\epsilon}$ and, for $j=1,\ldots,r$, let ${\hat \x}_j=\x_j\arrowvert_{y=\epsilon}$. Hence 
$$\w'=\sigma{\hat \x}_1\epsilon{\hat \x}_2\epsilon{\hat \x}_3\epsilon\cdots\epsilon {\hat \x}_r.$$

Now, if $\x_j=y^{\lambda_m}$, then ${\hat \x}_j=\epsilon^{\theta_j}$ where $\theta_j=\lambda_m$ with $\lambda_m$ is odd and so,
$\epsilon {\hat \x}_j\epsilon=\epsilon^{\theta_j+2}$. 

By \eqref{com:ni} and \eqref{com:nii}, if 
$\x_j=x^{\lambda_{l}}\epsilon y^{\lambda_{l+1}}\epsilon x^{\lambda_{l+2}}\epsilon y^{\lambda_{l+3}}\cdots \epsilon x^{\lambda_{m}}$
where $x$ is the unique proper letter in $\spt(\x_j) - \{y\}$, then ${\hat \x}_j=x^{\theta_j}$ with 
$\theta_j=\lambda_{l}+\lambda_{l+2}+\cdots+\lambda_{m}$ and so, $\epsilon {\hat \x}_j\epsilon=\epsilon x^{\theta_j} \epsilon$.

Therefore, 
\begin{equation}\label{eq:nuova2}
\w'=\sigma {\hat x}_{i_1}^{\theta_1}\epsilon {\hat x_{i_2}}^{\theta_2}\ldots\epsilon {\hat x}_{i_q}^{\theta_q},
\end{equation} 
where $q$ equals the number of $\x_j$'s such that $\spt(\x_j)\not=\{y\}$, $\{i_1\ldots,i_q\}$ is the image of an injection from 
$\{1,\ldots,q\}$ into $\{1,2,\ldots,r\}$ and the ${\hat x}_{i_h}$'s are proper letters in $\Sigma-\{y\}$ such that 
${\hat x}_{i_h}\not={\hat x}_{i_{h+1}}$, $h=1\ldots q-1$.
\mybreak
Suppose now that $y=x_{i_1}$. In this case set $\x_1=y^{\lambda_1}$ and let $\x_j$ be defined as above for $j\geq 2$. Since 
${\hat \x}_1=\epsilon^{\lambda_1}$ and $\lambda_1$ is odd by the parity conditions, it follows that 
${\hat \x}_1\epsilon=\epsilon^{\lambda_1+1}$ is the empty word in the pattern of $\w'$ and therefore $\sigma$ is followed by 
the first occurrence of $\hat{\x}_2$.
\mybreak
Finally, if  $y=x_{i_s}$, then set $\x_r=y^{\lambda_s}$ and $\x_j$ be defined as above for $j\leq r-1$. Thus $\epsilon {\hat \x}_r=\epsilon^{\lambda_s+1}$ 
and, again by the parity conditions, $\epsilon {\hat \x}_r$ is replaced by the empty word in the pattern of $\w'$. 
Therefore $\w'$ ends with the last occurence of $\hat{\x}_{r-1}$. The lemma is now completely proved.
\end{proof}
The following result provides the base step for the inductive proof of Theorem~\ref{thm:fund3}. 

\begin{lemma}
\label{lem:suff3}
If $\cw$ is a jump-free $s$-word with exactly two proper letters that satisfies the parity conditions, then $\cw$ is a sunword.
\end{lemma}

\begin{proof}
By jump-freeness, $\cw=[\w]$ where 
$\w=\sigma a^{\lambda_1}\epsilon b^{\lambda_2}\epsilon a^{\lambda_3}\epsilon b^{\lambda_4}\cdots \epsilon b^{\lambda_{s-1}}\epsilon a^{\lambda_s}$. 
By the parity conditions, $\lambda_1$ and $\lambda_s$ are both odd while all other $\lambda_i$'s are even. If 
suffices to check that $\cw\arrowvert_{a=\epsilon}$ and $\cw\arrowvert_{b=\epsilon}$ are both $s$-words (notice that there are no other possible 
projections).

By the definition of projection onto a set consisting of exactly one proper letter, it follows that 
\begin{itemize}
\item[]$\cw\arrowvert_{a=\epsilon}=\big[\w\arrowvert_{a=\epsilon}\big]=\big[b^{1+\lambda_2+\lambda_4+\cdots+\lambda_{s-1}}\big]$ and 
\item[]$\cw\arrowvert_{b=\epsilon}=\big[\w\arrowvert_{b=\epsilon}\big]=\big[a^{1+\lambda_1+\lambda_3+\cdots+\lambda_s}\big].$
\end{itemize}
Since $\cw$ satisfies the parity conditions, both $1+\lambda_2+\lambda_4+\cdots+\lambda_{s-1}$ and $1+\lambda_1+\lambda_3+\cdots+\lambda_s$ 
are odd and the thesis follows.
\end{proof}

\begin{theorem}
\label{thm:fund3} 
Let $\cw=[\w]$ be an $s$-word with at least two proper letters from $\Sigma$. If $\cw$ is jump-free and $\cw$ 
satisfies the parity conditions, then $\cw$ is a sunword.
\end{theorem}

\begin{proof}
The proof is by induction on $|\Sigma|=n$. If $n=2$ the thesis follows from Lemma~\ref{lem:suff3}. 
The inductive hypothesis is the following: every projection of an $s$-word on $n-1$ letters that is jump-free and satisfies the parity 
conditions is an $s$-word. We need to prove the same statement for an $s$-word $\cw$ on $n$ letters. 
First observe that the projections of $\cw$ consist of $\cw|_{y=\epsilon}$ for each $y\in \Sigma$ and the projections of $\cw|_{y=\epsilon}$ for each $y\in \Sigma$.
Now $\cw|_{y=\epsilon}$ is an $s$-word for each $y\in \Sigma$, by Lemma~\ref{lem:suff1}. 
So, if we prove that $\cw|_{y=\epsilon}$ is jump-free and satisfies the parity conditions for each proper letter $y$, then the
thesis will follow by inductive hypothesis because $\cw|_{y=\epsilon}$ is an $s$-word on $n-1$ letters. 
Since the property of being jump-free is inherited by projection, to prove the theorem it suffices to show that $\cw|_{y=\epsilon}$ satisfies 
the parity conditions for each proper letter $y$. 
\mybreak
Since $\cw$ is jump-free and satisfies the parity conditions, $\w$ has the form

\begin{center}
$\w=\sigma a^{\lambda_1}\epsilon b^{\lambda_2}\epsilon \vv \epsilon b^{\lambda_{s-1}}\epsilon a^{\lambda_s}$ 
with $\lambda_1$ and $\lambda_s$ odd, and $\vv$ is jump-free. 
\end{center}

Let $\cw'=\cw|_{y=\epsilon}$ and $\w'=\w|_{y=\epsilon}$. Hence $\cw'=[\pi(\w')]$.
Let $\sigma\prec \hat a\prec \hat b\prec \hat c\prec \cdots$ be the restriction of $\prec$ to $\Sigma-\{\epsilon,y\}$. 
By Lemma \ref{lem:suff1}, $\cw'$ is a jump-free $s$-word on $\Sigma-\{y\}$ and so, 
$\cw'=[\sigma \hat a^{\theta_1}\epsilon \hat b^{\theta_2}\epsilon{\hat \vv}\epsilon \hat b^{\theta_{r-1}}\epsilon \hat a^{\theta_r}],$ for some $r\leq s$ 
We now show that $\cw'$ satisfies the parity conditions.
\mybreak
Let us prove first condition \eqref{com:001} in Lemma~\ref{lemma:00}.
\mybreak
\textbf{Claim.}
{\it$\theta_1$ and $\theta_r$ are both odd.}
\mybreak
Write $\w$ as  $\sigma\x\epsilon\uu\epsilon\z$, where $\x$ and $\z$ are intervals of $\w$ having 
$\spt(\x)=\spt(\z)=\{\epsilon,a,b\}$ and maximal with this property. Hence $\x$ begins with $a$ and ends with $b$, 
$\z$ begins with $b$ and ends with $a$ and $\uu$ begins and ends with $c$, because $\cw$ is jump-free. 
So, for some $h$ and $k$ such that $2\leq h\leq k\leq s-1$, one has 
$\x=a^{\lambda_1}\epsilon b^{\lambda_2}\epsilon\cdots a^{\lambda_{h-1}}\epsilon b^{\lambda_h}$ and
$\z=b^{\lambda_k}\epsilon a^{\lambda_{k+1}}\epsilon\cdots b^{\lambda_{s-1}}\epsilon a^{\lambda_s}$. 

\noindent
By the parity conditions 
\begin{equation}
\label{eq:paritycons}
\text{$\lambda_1$, $\lambda_h$, $\lambda_k$ and $\lambda_s$ are odd while all the other exponents in $\x$ and $\y$ are even.}
\end{equation}
Denote by $\x'$, $\z'$ and $\uu'$ the projection of $\x$, $\z$ and $\uu$, respectively, onto $\Sigma-\{y\}$.
Now $\w'=\sigma\x'\epsilon\uu'\epsilon\z'$. 
If $c\preceq y$, then $\hat a=a$, $\hat b=b$. As a consequence, $\x'=\x$ and $\z'=\z$ and so, $\theta_1=\lambda_1$ and 
$\theta_r=\lambda_s$. Therefore $\theta_1$ and $\theta_r$ are both odd. 

Suppose that $y=b$. In this case, $\hat a=a$, $\hat c=c$, $\x'=a^{\lambda_1+\lambda_3+\cdots+\lambda_{h-1}}\epsilon^{\lambda_h}$ and $\z'=\epsilon^{\lambda_k}a^{\lambda_{k+1}+\lambda_{k+3}+\cdots+\lambda_s}$. Hence
$$\pi(\w')=\sigma a^{\lambda_1+\lambda_3+\cdots+\lambda_{h-1}}\epsilon\uu'\epsilon a^{\lambda_{k+1}+\lambda_{k+3}+\cdots+\lambda_s}=\sigma a^{\theta_1}\epsilon\uu'\epsilon a^{\theta_r}$$
with $\theta_1$ and $\theta_r$ both odd because of \eqref{eq:paritycons}.
\mybreak
Analogously, if $y=a$, then $\hat a=b$, $\hat b=c$, $\x'=\epsilon^{\lambda_1+1}b^{\lambda_2+\lambda_4+\cdots+\lambda_{h}}$ and $\z'=b^{\lambda_k+\lambda_{k+2}+\cdots\lambda_{s-1}}\epsilon^{\lambda_s+1}$. Hence 
$$\pi(\w')=\sigma b^{\lambda_2+\lambda_4+\cdots+\lambda_h}\epsilon\uu'\epsilon b^{\lambda_k+\lambda_{k+2}+\cdots+\lambda_{s-1}}=\sigma \hat a^{\theta_1}\epsilon\uu'\epsilon 
\hat a^{\theta_r}$$
with $\theta_1$ and $\theta_r$ both odd because of \eqref{eq:paritycons}. This shows that the parity condition (i) in Lemma~\ref{lemma:00}
is satisfied by $\cw'$. 
\hfill ({\it End of Claim}) 
\mybreak
Finally, we prove that $\cw'$ satisfies the parity condition \eqref{com:002} in Lemma~\ref{lemma:00}. To this end consider an interval of 
$\pi(\w')$ of the form  
$t^{\theta_{i-1}}\epsilon u^{\theta_i}\epsilon v^{\theta_{i+1}}$ where $t$, $u$ and $v$ are proper letters in $\Sigma-\{y\}$. As in the proof 
of Lemma \ref{lem:suff1}, we know that such an interval is the image (under projection and under $\pi$) of the interval 
$\z=\tw\kk_1 \uu \kk_2\vv\lhd \w$ where
\begin{itemize}
\setlength{\itemsep}{-1mm}
\setlength{\topsep}{0pt}
\setlength{\partopsep}{0pt}
\item[--] $\tw,\uu,\vv\lhd \w$ are such that $t\in \spt(\tw)\subseteq \{\epsilon,t,y\}$, $u\in \spt(\uu)\subseteq \{\epsilon,u,y\}$, $v\in \spt(\vv)\subseteq \{\epsilon,v,y\}$; hence $\tw\arrowvert_{y=\epsilon}=t^{\theta_{i-1}}$, $\uu\arrowvert_{y=\epsilon}=u^{\theta_i}$ and $\vv\arrowvert_{y=\epsilon}=v^{\theta_{i+1}}$; 
\item[--] $\kk_1,\kk_2\lhd \w$ are such that $\epsilon \in \spt(\kk_h)\subseteq \{\epsilon,y\}$ for $h=1,2$; hence either $\kk_h=\epsilon$ or $\kk_h=\epsilon y^\lambda\epsilon$ for some $\lambda$ odd. In any case $\kk_h\arrowvert_{y=\epsilon}=\epsilon$ for $h=1,2$. 
\end{itemize}

We now show that the parity conditions hold knowing that such conditions hold for $\z$. Let us first rule out the easiest case, namely when 
$y\not\in \spt({\uu})$. In this case, $\uu=u^{\lambda_l}$ for some $l\geq i$ and $\kk_h=\epsilon$, $h=1,2$. Hence $\theta_i=\lambda_l$. 
Since $\w$ satisfies the parity conditions, $\lambda_l$ (and hence $\theta_i$) is even or odd according to whether or not  $u$ is interlaced 
by the same letter. But $u$ is interlaced by the 
same letter or not according to whether or not $\spt(\tw)-\{\epsilon,y\}=\spt(\vv)-\{\epsilon,y\}$. Hence if $t=v$, then $\theta_i$ 
is even otherwise $\theta_i$ is odd and the parity conditions hold in this case. 

We therefore assume throughout the rest of the proof that $y\in\spt(\uu)$. 

\smallskip\noindent
{\bf Case~1.} {\it If $t=v$, then $\theta_i$ is even. } 

       Since $t=v$ and $\cw$ is jump-free, only two cases may occur, namely either $t\preceq y \preceq u$ or 
       $u\preceq y \preceq t$. Possibly by replacing $\w$ by $-\w$ we may suppose without loss of generality that 
       $t\preceq y \preceq u$. Hence $\kk_1=\epsilon y^{\lambda_l}\epsilon$ for some $l$ and $\kk_2=\epsilon y^{\lambda_m}\epsilon$ for some 
       $m\geq l+2$ and $\z=\tw \epsilon y^{\lambda_l}\epsilon \uu\epsilon y^{\lambda_m} \epsilon \tw$. We distinguish two cases.

\begin{itemize}
\setlength{\itemsep}{-1mm}
\setlength{\topsep}{0pt}
\setlength{\partopsep}{0pt}
\item[-] If $\uu=u^{\lambda_{l+1}}$, then $\theta_i=\lambda_{l+1}$, $m=l+2$ and since $\cw$ satisfies the parity conditions $\lambda_{l+1}$ (and, consequently, $\theta_i$) is even $u$ being interlaced by the same letter proper letter $y$. So we are done in this case.
\item[-] If $\uu=u^{\lambda_{l+1}}\epsilon y^{\lambda_l}\epsilon\cdots \epsilon u^{\lambda_{m-1}}$, then $\theta_i$ has the same parity as $\lambda_{l+1}+\lambda_{m-1}$ (by \eqref{eq:nuova1} and \eqref{com:nii}).
Moreover, since both the first and the last occurence of $u$ in $\uu$ are interlaced by $y$, $\lambda_{l+1}+\lambda_{m-1}$ (and, consequently, $\theta_i$) is even. 
\end{itemize}

\smallskip\noindent
{\bf Case~2.} {\it If  $t\not=v$, then $\theta_i$ is odd.} 

As before, under the assumptions of the theorem, only two cases may occur, namely $t\preceq y \preceq u\preceq v$ or $t\preceq u\preceq y \preceq v$ and, by the same reasons given above, we may assume without loss of generality, that $t\preceq y \preceq u\preceq v$. Hence $\kk_1=\epsilon y^{\lambda_l}\epsilon$ for some $l$ and $\kk_2=\epsilon$. Accordingly, $\z=\tw \epsilon y^{\lambda_l}\epsilon \uu\epsilon \vv$. Again, we distinguish two cases.
\begin{itemize}
\setlength{\itemsep}{-1mm}
\setlength{\topsep}{0pt}
\setlength{\partopsep}{0pt}
\item[-]If $\uu=u^{\lambda_{l+1}}$, then $\theta_i=\lambda_{l+1}$; since $\cw$ satisfies the parity conditions it follows that $\lambda_{l+1}$ (and, consequently, $\theta_i$) is odd $u$ being interlaced by two different proper letters, namely, $y$ and the first occurrence of $\vv$ which is $v$. So we are done in this case.
\item[-] If $\uu=u^{\lambda_{l+1}}\epsilon y^{\lambda_l}\epsilon\cdots \epsilon u^{\lambda_{m-1}}$, then $\theta_i$ has the same parity as $\lambda_{l+1}+\lambda_{m-1}$, where $\lambda_{m-1}$ is the exponent of the last occurrence of $u$ in $\uu$ (still by \eqref{eq:nuova1} and \eqref{com:nii}). By the parity conditions for $\cw$, $\lambda_{l+1}$ is even because it is the exponent of the first occurrence of $u$ in $\uu$ and such an occurrence is interlaced by the same letter $y$ (by \eqref{com:nii}). On the other hand $\lambda_{m-1}$ is odd because the last occurrence of $u$ in $\uu$ is interlaced by $y$ and the first occurrence of $\vv$ which is $v$. Therefore $\lambda_{l+1}+\lambda_{m-1}$ (and, consequently, $\theta_i$) is odd.
\end{itemize}
The proof is thus completed.
\end{proof}

The results of Theorem~\ref{thm:fund2} and Theorem~\ref{thm:fund3} can be equivalently and simultaneously stated as follows.
\begin{theorem}\label{thm:fund0}
An $s$-word on an alphabet $\Sigma$ is a sunword if and only if it is jump-free and satisfies the parity conditions.
\end{theorem}

\section{Minimally unbalanced diamond-free graphs}
\label{sec:mnb}

\mybreak
We finally come back to the problem of characterizing HOH-free multisuns and 
we show that sunoids are precisely the HOH-free multisuns. 
In other words we prove that the multisuns that hereditarily satisfy the N-conditions are hereditarily odd hole free. 
This result joint with Theorem~\ref{thm:fund1} provides a complete characterization of HOH-multisuns.

\begin{remark}\label{rem:1sunoid}
The $N$-conditions are sufficient to characterize HOH-free multisuns with exactly one inscribed clique.
In fact, let $G$ consist of an odd cycle $C$ with an odd clique $A$ inscribed so that each $A$-path has even order. 
Then $G$ does not contain odd holes and since the only submultisun of $G$ is $G$ itself it follows trivially that $G$ is HOH-free. 
\end{remark}

\begin{lemma}
\label{lemma:sunhole}
Let $\Gamma$ be a hole in a sunoid $G$. Then there exists a sub-multisun $G_0$ of $G$ with $q$ inscribed cliques $X_1\ldots,X_q$, a permutation $i_1,i_2\ldots,i_q$ of $\{1,\ldots,q\}$ and $q$ vertex-disjoint paths $P_1,\ldots,P_q$ such that
\begin{itemize}
\item[--] $P_h$ is a $X_{i_i}X_{i_{h+1}}$-path of $G_0$ for $h=1,\ldots,q-1$ and $P_q$ is $X_1X_q$-path of $G_0$. 
\item[--] $V(\Gamma)=V(P_1)\cup V(P_2)\cup\cdots\cup V(P_{q-1})\cup V(P_q)$.
\end{itemize}
\end{lemma}
\begin{proof}
Let $G$ be a sunoid with rim $G$ and let $\Gamma$ be a hole in $G$. Among all those sub-multisuns of $G$ containing $\Gamma$ let $G_0$ 
have the least possible number of inscribed cliques. Let $X_1\ldots X_q$ be the inscribed cliques of $G_0$ and let $F_1,\ldots, F_q$ be 
the corresponding edge-sets. If $q=1$ then we are done, because the vertex-sets of holes in sunoids with exactly one inscribed cliques $A$, 
are precisely the $A$-paths. We therefore suppose that $q\geq 2$. Since $G$ is a sunoid, then $G_0$ is such. In particular $G_0$ satisfies 
N-\ref{com:p3}. Let $\xi$ be the common vertex to $X_1,\ldots,X_q$ and let $f_0$ be the canonical labeling of $G_0$ where $\Sigma_0=f_0(V(C))=\{\epsilon,\sigma,x_1\ldots,x_q\}$. We say that $\Gamma$ uses label $x$ if some vertex of $\Gamma$ is labeled $x$. We notice that:
\begin{itemize}
\item[--] $\#(E(\Gamma)\cap F_i)\leq 1$ because $\Gamma$ is a hole and $X_i$ is a clique;
\item[--] $E(\Gamma)\cap F_i=\emptyset$ for no $i\in \{1,\ldots,q\}$ because if $E(\Gamma)\cap F_i $ were empty for some $i\in \{1,\ldots,q\}$, then $E(\Gamma)\subseteq G_0-F_i$ and $\Gamma$ would be an odd hole in a sub-multi-sun with less inscribed cliques than $G_0$ contradicting the choice of $G_0$.
\item[--] $\sigma$ is not used by $\Gamma$, because, since $\xi$ is adjacent to all properly labeled vertices, $\Gamma$ could not be a hole.
\end{itemize}
We therefore conclude that $\#(E(\Gamma)\cap F_i)=1$ for each $i\in \{1,\ldots,q\}$. 
Moreover, since $(X_i-\xi)\cap (X_j-\xi)=\emptyset$ if $i\not=j$, $i,\,j\in  \{1,\ldots,q\}$, it follows that the edges of $\Gamma$ 
contained in the inscribed cliques of $G_0$ are not adjacent and, therefore, each proper label of $\Sigma_0$ is used exactly twice. 
Thus all vertices of $U=V(\Gamma)-f_0^{-1}(\Sigma_0-\{\epsilon,\sigma\})$ are labeled $\epsilon$. 

Let $\{e_i\}=E(\Gamma)\cap F_i$ for each $i\in \{1,\ldots,q\}$ and orient $\Gamma$. The orientation of $\Gamma$ induces a permutation $i_1 i_2 i_3\ldots i_q$ of 
$\{1,2\ldots,q\}$ as follows: $i_k=j$ if $e_j$ is $k$-th edge not on the rim we meet traveling along $\Gamma$ according to the prescribed 
orientation. 
\mybreak

The crucial observation is that $U$ induces $q$ vertex disjoint path each one being the interior of some $X_{i_h}X_{i_{h+1}}$-path of $G_0$. 
To see this, consider the graph $\Gamma-\{e_1,\ldots,e_q\}$. Since $\Gamma$ is a hole and the $e_i$'s are pairwise nonadjacent, 
$\Gamma-\{e_1,\ldots,e_q\}$ consists of $q$-paths pairing the endpoints of the $e_i$'s. 
Since vertices with the same proper label are paired in $\Gamma$ by the $e_i$'s, the pairing of the endpoints of the $e_i$'s in 
$\Gamma-\{e_1,\ldots,e_q\}$ consists of $q$ paths pairing vertices with different proper labels. It follows that such $q$ paths 
are paths connecting vertices in different inscribed cliques of $G_0$ and such that each vertex of their interior is labeled $\epsilon$. 
By the definition of $AB$-paths, each of the latter paths is precisely  an $X_{i_h}X_{i_{h+1}}$-path for some $h=1,\ldots, q$ (sums are 
taken mod $q$). Such paths are moreover, vertex disjoint, and the union of their vertex sets is $V(\Gamma)$ as required.
\end{proof}

\comment{

\begin{remark}
It is worth noticing that if $G$ is a sunoid and $\hat{G}$ is the standard multisun of $\mathcal{S}_G$, then the holes of $\hat{G}$ have order that is a multiple of 3. Indeed the order of any $XY$-path is 3. Hence the vertices of any hole $\Gamma$ of $\hat{G}$ are $3q$, by the lemma.
\end{remark}
}

In the language of $s$-words, the statement of Lemma \ref{lemma:sunhole} is equivalent to the following.

\begin{lemma}
\label{lemma:sunhole1}
Let $G$, $G_0$, $\Sigma_0$ and $\Gamma$ as in Lemma \ref{lemma:sunhole} and let $q$ be the number of inscribed cliques of $G_0$. Let $\cw_0$ 
be the $s$-word of $G_0$ and let $\w_0$ be any representative of $\cw_0$ coinciding with is own pattern. Then, either $q=1$ and $\Gamma$ 
has even order, or $\Gamma$ is represented by a sequence of pairs 

$$(x_{i_1},x_{i_2})(x_{i_2},x_{i_3})\ldots (x_{i_{q-1}},x_{i_q})(x_{i_q},x_{i_1})$$

such that, for each $h=1\ldots, q$, either $x_{i_h}^{\eta_h}\epsilon x_{i_{h+1}}^{\eta_{h+1}}$ or $x_{i_{h+1}}^{\eta_{h+1}}\epsilon x_{i_h}^{\eta_h}$ is an interval of $\w_0$ for some integers 
$\eta_1\ldots\eta_q$ (sums are taken modulo $q$). 
\end{lemma}

\begin{proof}
For $i=1,\ldots,q$ let $u_i$ and $v_i$ be the end-vertices of $e_i$ where $e_i$ is defined in the proof of Lemma \ref{lemma:sunhole},
i.e., $\{e_i\}=E(\Gamma)\cap F_i$. 
Recall that $u_i$ and $v_i$ are both labeled $x_i$, $i=1\ldots,q$. The orientation of $\Gamma$ induces an orientation of each $e_i$, 
$i=1,\ldots,q$. Suppose without loss of generality that the orientation chosen for $\Gamma$ orients $e_1$ from $u_1$ to $v_1$. 
Hence, by Lemma \ref{lemma:sunhole}, for $h=1\ldots,q$ the path $P_h$, occurring as one of the vertex disjoint paths factorizing $V(G_0)$, 
is an $X_{i_h}X_{i_{h+1}}$-path of $G_0$. 
Hence, traversing $\Gamma$ according to the prescribed orientation and starting from $v_1$, we encounter 
the $X_{i_h}X_{i_{h+1}}$-path $P_h$ of $G_0$ and we represent it as a pair $(x_{i_h},x_{i_{h+1}})$;
then, each time we reach a vertex $u_{i_{h+1}}$ labeled $x_{i_{h+1}}$ we walk through $e_{h+1}$ and reach 
the vertex $v_{i_{h+1}}$ of $\Gamma$ that is labeled $x_{i_{h+1}}$, for $h=1,\ldots,q-2$. 

Since $P_h$ is an $X_{i_h}X_{i_{h+1}}$-path in $G_0$, it follows that either $x_{i_h}^{\eta_h}\epsilon x_{i_{h+1}}^{\eta_{h+1}}$ or $x_{i_{h+1}}^{\eta_{h+1}}\epsilon x_{i_h}^{\eta_h}$ is an interval of $\w_0$, for some $\eta_1\ldots\eta_q$,  
\end{proof}

We are now ready to state and prove the characterizations of HOH-free multisuns and, equivalently, of minimally unbalanced diamond-free graphs. 
The reader should recall that such graphs are those graphs whose clique-matrix is minimally non-balanced in the sense that, removing any 
vertex from the graph the resulting clique-matrix is balanced while the clique-matrix of the graph is not balanced. 

\begin{theorem}\label{thm:summarizing}
Let $G$ be a multisun. The following statements are equivalent
\begin{enumerate}[{\rm (1)}]
\item\label{com:summ1} $G$ is a sunoid.
\item\label{com:summ2} The $s$-word of $G$ is a sunword.   
\item\label{com:summ3} $G$ is HOH-free.
\end{enumerate}
\end{theorem}  
\begin{proof}
We have already shown that \eqref{com:summ1}$\Leftrightarrow$\eqref{com:summ2}. Since HOH-free multisuns satisfies the N-conditions 
hereditarily, we have that \eqref{com:summ3}$\Rightarrow$\eqref{com:summ1}. It remains to show that \eqref{com:summ2}$\Rightarrow$
\eqref{com:summ3}. To this end let $\cw$ be the $s$-word of $G$. Hence $G$ is in $\mathcal{S}_{G_\cw}$. To prove that $G_\cw$ is HOH-free it suffices 
to prove that $G_{\cw'}$ id odd hole free for any projection $\cw'$ of $\cw$. On the other hand, since any projection of a sunword is a 
sunword it suffices to prove that $G_\cw$ itself is odd hole-free or, equivalently, that each hole $\Gamma$ of $G_\cw$ has even order. 
Since $\cw$ is a sunword, then Lemma \ref{lemma:sunhole1} applies and if $G_\cw$ contains a hole $\Gamma$, then $\Gamma$ is contained 
in some sub-multisun $G_0$ with $q$ inscribed cliques. We show that $\Gamma$ has indeed even order. 

By Lemma \ref{lemma:sunhole}, we may assume that $q\geq 2$, otherwise we are done. Hence, by Lemma~\ref{lemma:sunhole1}, 
$\Gamma$  is represented by the sequence $(x_{i_1},x_{i_2})(x_{i_2},x_{i_3})\ldots(x_{i_{q-1}},x_{i_q})(x_{i_q},x_{i_1})$ and  
there is a word $\w_0$ on $\{\epsilon,\sigma,x_1\ldots, x_q\}\subseteq\Sigma$ 
such that $[\w_0]$ is a projection of $\cw$.
Notice that $G_0$ is in $\mathcal{S}_{G_{\cw_0}}$, where $\cw_0=[\w_0]$. 
Clearly $i_1 i_2 i_3\ldots i_q$ is a permutation $\rho$ of $1 2 3\cdots q$. Let $\preceq$ be the linear order induced by 
$\cw_0$. 
We may suppose that $x_1\preceq x_2\preceq x_3\ldots\preceq x_q$ possibly by re-labeling (in this case the permutation $\rho$ 
changes accordingly). 
\mybreak
We claim that $q=2$.
For, if $q>2$, then there is at least one index $m\in\{1,2,\ldots q\}$  such that $|\rho^{-1}(i_m)-\rho^{-1}(i_{m-1})|>1$. Hence, for 
some $m$, $x_{i_m}$ and $x_{i_{m+1}}$ do not form
a cover pair in $\preceq$. Since, by Lemma~\ref{lemma:sunhole1}, either $x_{i_h}^{\eta_h}\epsilon x_{i_{h+1}}^{\eta_{h+1}}\lhd \w_0$ or $x_{i_{h+1}}^{\eta_{h+1}}\epsilon x_{i_h}^{\eta_h}\lhd \w_0$ for $h=1\ldots, q\, \imod{q}$, we conclude that 
$\cw_0$ contains a jump involving $x_{i_m}$ and $x_{i_{m+1}}$, contradicting that $\cw_0$ is jump-free. Therefore $q\leq 2$ and, 
consequently, $q=2$ because we are assuming $q\geq 2$.
\mybreak
Now, $q=2$ implies that $V(\Gamma)=P_1\cup P_2$ for some two vertex disjoint $X_{i_1}X_{i_2}$-paths of $G_0$. 
Since $P_1$ and $P_2$ have the same odd parity by N-\ref{com:p5}, we conclude that $\Gamma$ has even order. 
Therefore we proved that each member of $\mathcal{S}_{G_\cw}$, in particular $G$, is HOH-free.
\end{proof}

As a consequence of the previous theorem, sunoids, HOH-free multisuns and sunwords are (essentially) the same thing. 
In particular, minimally unbalanced diamond-free graphs that are not odd holes are precisely the multisuns 
whose $s$-word is a sunword. 
Since we know how to build and recognize sunwords, we also know how to build and recognize HOH-free multisuns and, 
consequently, the structure of minimally unbalanced diamond-free graphs. Indeed, we have the following.

\begin{coro}\label{coro:utile}
Let $G$ be a diamond free graph. Then $G$ is not balanced if and only if it contains either an odd hole or a sunoid as an 
induced subgraph. 
\end{coro}

\section{Consequences}
\label{sec:consequence}

In this section we briefly consider some consequences of our result.

\subsection{Algorithmic consequence}

Sunoids can be used to provide a graph-theoretical interpretation of the algorithm of Conforti and Rao designed in \cite{CR92a} to recognize balanced matrices assuming the existence of a routine testing perfection. We show that our characterization implies exactly the same algorithm when specialized to linear matrices. 

\begin{prop}\label{prop:oucharact}
Let $\A$ be a linear matrix. Then $\A$ is balanced if and only if $\A$ does not contain any submatrix congruent to $\C_3$ and $G_\A$ is a 
diamond-free graph that contains neither an odd-hole nor an HOH-free multisun as an induced subgraph. 
\end{prop}
\begin{proof}
If $\A$ is a balanced matrix, then $\A$ does not contain submatrices congruent to $\C_n$ for any odd integer $n$ and, in particular, $\A$ 
does not contain submatrices congruent to $\C_3$. Since $\A$ is linear $G_\A$ is diamond-free and since $\A$ is balanced, so is $G_\A$. 
Hence, the necessity follows by Corollary \ref{cor:2sec2}. Let us prove the sufficiency. By Lemma~\ref{lemma:AB}, $\A$ is conformal 
because $\A$ is a linear matrix that does not contain any submatrix congruent to $\C_3$. 
Therefore $\A^\uparrow\cong \A_{G_{\A}}$. Since by Corollary \ref{cor:2sec2}, $G_\A$ is a balanced graph, it follows that $\A^\uparrow$ is a 
balanced matrix being the clique matrix of a balanced graph. To complete the proof it suffices to show that the following two statements 
about a linear matrix are equivalent
\begin{enumerate}[(i)]
\setlength{\itemsep}{-1mm}
\setlength{\topsep}{0pt}
\setlength{\partopsep}{0pt}
\item\label{com:zi} $\A$ is balanced;
\item\label{com:zii} $\A^\uparrow$ is balanced;
\end{enumerate}
Clearly \eqref{com:zi}$\Rightarrow$\eqref{com:zii}. The fact that \eqref{com:zi}$\Rightarrow$\eqref{com:zii} follows at once by the 
following remark due straightforwardly to the linearity of $\A$:

{\em the rows of $\A$ which are not copies of rows of $\A^\uparrow$ have exactly one nonzero entry}.

\noindent
Hence if $\A$ contains an odd cycle submatrix so does $\A^\uparrow$. Therefore if $\A^\uparrow$ is balanced so is $\A$. 
\end{proof} 

\begin{remark}
\label{rem:linbal}
In the proposition above we showed that when $\A$ is linear, then $\A$ is balanced if and only if $\A^\uparrow$ is such. This fact is 
no longer true for general matrices. Take for instance, ${\C_3 \brack 1\,1\,1}$. Then ${\C_3 \brack 1\,1\, 1}^\uparrow=[1\,1\,1]$ is balanced while 
${\C_3 \brack 1\,1\,1}$ is not. 
\end{remark}

We can now describe the algorithm whose correctness relies on the next proposition.
\mybreak
\mybreak
Let $\A$ have $n$ columns and let $R_j$ be the indices of the rows of $\A$ that intersect column $j$ of $\A$ in a 1. Let $\#R_j=m_j$ and denote 
by $\mathcal{A}^{(j)}$ the class of submatrices of $\A$ obtained by removing a set of $m_j-2$ rows from $\A$. Furthermore, denote by 
$\mathcal{G}^{(j)}$ the class of graphs $\big(G_{\A'} \ |\ \A'\,\text{is in}\,\, \mathcal{A}^{(j)}\big)$.

\begin{prop}\label{prop:algorithm}
Let $\A$ be a linear matrix with $g(\A)\geq 5$. If $G_\A$ does not contain odd holes, then $\A$ is balanced if and only if 
no member of $\mathcal{G}^{(j)}$ contains odd holes.
\end{prop}
\begin{proof}
Since $g(\A)\geq 5$, it follows $\A$ is conformal by Lemma~\ref{lemma:AB}. By Corollary~\ref{cor:2sec2} and Proposition~\ref{prop:oucharact}, 
$\A$ is balanced if and only if $G_\A$ does not contain odd holes and HOH-free multisuns. Therefore, under the 
assumptions, $\A$ is not balanced if and only $G_\A$ contains some HOH-free multisun. Suppose that $G_\A$ contains an induced HOH-free multisun 
$H$ with rim $C$. The unique vertex labeled $\sigma$ in $H$ corresponds to the $j$-th column of $\A$. Let $\mathcal{K}^{(j)}$ be the set of maximal 
cliques of $H$ containing $\sigma$. 
The cliques in $\mathcal{K}^{(j)}$ are the inscribed cliques of $H$ plus the two edges of the rim, $e$ and $f$, say, 
incident to $j$. Therefore, the graph $H'$ obtained by removing from $H$ the edge-sets of the maximal cliques in $\mathcal{K}^{(j)}-\{e,f\}$, 
is an odd hole and $H'$ is a subgraph of some member of $\mathcal{G}^{(j)}$. We conclude that $G_\A$ contains an induced 
HOH-free multisun if and only if, for some $j=1,\ldots,n$, some member of $\mathcal{G}^{(j)}$ contains an odd hole.
\end{proof}

\begin{remark}\label{rem:holesindf}
It is worth noticing that for diamond-free graphs, odd hole-freeness is equivalent to perfectness.
\end{remark}

The algorithm is now rather trivial. It consists of polynomially many calls to a routine \texttt{Odd\_Holes} that tests whether a given 
diamond-free graph contains an odd hole: it returns \texttt{YES} if an odd hole is found and \texttt{NO} otherwise. Since the latter problem is 
solvable in polynomial-time, so is the problem of testing whether a given linear matrix or (equivalently) a diamond-free graph is balanced.
Indeed let $\A$ be a linear matrix. Check first if $g(\A)\geq 5$. If not $\A$ is not balanced, else $\A$ is conformal. Call 
\texttt{Odd\_Holes} on $G_\A$. If the routine returns \texttt{YES}, then the algorithm stops with the declaration that $\A$ is not balanced. 
Else, for $j=1,\ldots, n$ one calls \texttt{Odd\_Holes} on each of the ${m_j\choose m_j-2}={m_j\choose 2}$ members of $\mathcal{G}^{(j)}$. 
If all such tests fail the matrix $\A$ is balanced.

\subsection{Dyck-paths} 

As mentioned throughout the paper sunoids exhibit a large amount of geometrical structure. We show now that, rather surprisingly, sunoids and 
hence balanced linear matrices and balanced diamond-free graphs, have intimate relationships with other objects in enumerative combinatorics.
\mybreak
A \emph{Dyck-path} is a lattice path in $\mathbb{R}^2$ whose points $P_0,P_1,\ldots P_{2n}$ have nonnegative ordinates and satisfy the following 
relations: 
 
$$P_0=(0,0), \,\ P_{2n}=(0,2n), \,\ P_{i+1}-P_i\in \{(1,-1),(1,1)\}.$$ 

The number $n$ is the \emph{semi-length} of the Dyck-path. Hence, one might think of this path as evolving under the following rules: 
it starts at the origin, it ends in $(0,2n)$ and if it reaches a point $P_i$, then it moves to the next point 
$P_{i+1}$ by either a ``down-step''  or an ``up-step'' so that it never falls beyond the $x$-axis. 

Dyck-paths are represented by {\em Dyck-words}, namely words on a two letters alphabet, say $\{L,R\}$, such that no prefix of the word has 
more $R$’s than $L$’s. For instance, $RRRLLL$, $LRLLRR$, $LRLRLR$, $LLRRLR$, $LLRLRR$ are the Dyck-words of length 6. 
If one thinks of the symbol $L$ as an open parenthesis and of $R$ as a closed parenthesis, Dyck-words 
represents expressions with $n$ pairs of parentheses that are correctly matched also known as \emph{legal bracketings} in \cite{BM96}.

Dyck-words of length $2n$ are enumerated by the Catalan number $C_n$ so that, agreeing with the most common definition of Catalan number, 
$C_n$ is the number of way of pairing $2n$ parentheses. 
The bijection between Dyck-words and Dyck-paths is established as follows: let $\vv=v_1\cdots v_{2n}$ 
be a Dyck-word on $\{L,R\}$ and for $i=1,\ldots,n$, let $h_i$ be the differences between the occurrences of $L$ and those of $R$ in the 
prefix $v_1v_2\cdots v_i$. 
Then $\{(0,0)\}\cup\{(i,h_i) \ |\ i=1\ldots,2n\}$ is a Dyck-path. Conversely, if $\{P_0,P_1,\ldots P_{2n}\}$ is a Dyck-path, then by setting $v_i=L$ is 
$P_i-P_{i-1}$ is an up-step and $v_i=R$ if $P_i-P_{i-1}$ is an down-step, $i=1\ldots,2n$, we associate a Dyck-word with the path.
\mybreak
We now show how sunwords and Dyck-paths are related. For $i=1,\ldots,2n-1$, a point $P_i$ of a Dyck-path is a \emph{peak} if $P_{i-1}$ and 
$P_{i+1}$ have the same ordinate. An \emph{evenly weighted Dyck-path $D$} of semilength $n$ is a pair $(D,\Lambda)$ where $D$ is a Dyck-path 
of semilength $n$ and $\Lambda: \{0,1\ldots,2n\}\rightarrow \mathbb{Z}_+$ is a mapping such that $\Lambda(i)$ is even if $P_i$ is a peak and 
$\Lambda(i)$ is odd otherwise, $i=0,\ldots,2n$. 
Clearly every Dyck-path $D$ identifies an entire set of weights $\Lambda$ such that $(D,\Lambda)$ is an evenly weighted Dyck-path $D$.

Let now $\cw=[\sigma z_{i_1}^{\lambda_1} \epsilon z_{i_2}^{\lambda_2}\epsilon z_{i_3}^{\lambda_3}\cdots\epsilon z_{i_{s}}^{\lambda_{s}}]$ be an $s$-word and let 
$\preceq$ be the order induced by $\cw$. By Proposition~\ref{prop:transit}, we know that $s$ is odd. Moreover, we may choose the 
representative of $\cw$ so that $z_{i_1}=a$. For $j=1\ldots s$, let $h(z_{i_j})$ be the rank of $z_{i_j}$ in $\preceq$. 
Thus $h(a)=1$, $h(b)=2$, $h(c)=3$ and so on. Moreover, for $j=0,\ldots, s-1$ let $P_j=(j,h(z_{i_{j+1}}))$. Thus $D_\cw=\{P_0,P_1,\ldots P_{s-1}\}$ is 
a set of lattice points of $\mathbb{R}^2$ with $P_0=(0,0)$ and $s-1$ even. Finally let $\Lambda_\cw: \{0,1,\ldots,s-1\}\rightarrow \mathbb{Z}_+$ 
be defined by $\Lambda_\cw(j)=\lambda_{j+1}$.

\begin{theorem}\label{thm:dyck}
Let $\cw$ be an $s$-word with at least two proper letters. Then $\cw$ is a sunword if and only if $(D_\cw,\Lambda_\cw)$ is an evenly weighted 
Dyck-path.
\end{theorem}
\begin{proof}
It suffices to observe that the fact that $\cw$ is jump-free is equivalent to the fact that $D_\cw$ is a Dyck-path and that the fact that 
$\cw$ satisfies the parity conditions is equivalent to the fact that $\Lambda_\cw$ is such that $D_\cw$ is evenly weighted.
\end{proof}

\subsection{Clique-perfection}

A {\em clique-transversal} of a graph $G$ is a subset of vertices that meets all the maximal cliques of $G$. 
A {\em clique-independent set} of a graph is a collection of pairwise 
vertex-disjoint cliques of $G$. We denote the minimum size of a clique-transversal and the maximum size of 
a clique-independent set by $\tau_c(G)$ and $\alpha_c(G)$, respectively. These graph-invariants were introduced in 
\cite{Tu90,GP00} and since then many authors studied their mutual relations \cite{BCDII07, BCD97,LC06,LT86}. 
\mybreak
A graph $G$ is {\em clique-perfect} if  

\begin{equation}
\label{eq:cliqueperfection}
\tau_c(G)=\alpha_c(G)\qquad\text{for each induced subgraph $G'$ of $G$.}
\end{equation}


Graphs that seem to play a key role in the characterization of clique-perfect graphs are the so-called {\em suns} 
and their generalizations. 
\mybreak
A {\em sun} is a chordal graph $G$ whose vertex set can be partitioned into two sets
$W=\{w_1,\dots, w_r\}$ and $U=\{u_1,\dots,u_r\}$ such that $U$ is a stable set and for each $i$ and $j$, $w_j$ is
adjacent to $u_i$ if and only if $j=i $ or $j=i+1\imod{r}$. A sun is {\em odd} if $r$ is odd and is {\em complete} if $W$ is complete. Given a cycle $C$, the edges of $C$ that form a triangle with another vertex of $C$ are called {\em non-proper}. 
An {\em odd generalized sun} is a graph $G$ whose vertex set can be partitioned into two sets: a (not 
necessarily induced) odd cycle $C$ of $G$ with non-proper edges $\{e_j\}_{j\in J}$ ($J$ is allowed to be empty) 
and a stable set $U=\{u_j\}_{j\in J}$ such that $u_j$ is adjacent only to the endpoints of a non-proper edge of $C$. 
Clearly odd holes and odd suns are odd generalized suns and all these graphs are not clique-perfect  \cite{BDGS06}. 

Bonomo, Chudnovsky and Duran \cite{BCDII07} gave the following partial characterization of diamond-free clique-perfect graphs 
in terms of (not minimally) induced subgraphs. 

\begin{theorem} [\cite{BCDII07}]
\label{theo:first_char}
Let $G$ be a diamond-free graph. Then $G$ is clique-perfect if and only if no induced subgraph of $G$ is an
odd generalized sun. 
\end{theorem}

Unfortunately, not every odd generalized sun is minimally clique-imperfect (with respect to induced 
subgraphs). For instance, an odd hole with an inscribed clique of even size is not clique-perfect but 
contains an odd hole as a minimal clique-imperfect subgraphs. This implies that the previous characterization is not minimal.

Our characterization on minimally unbalanced diamond-free graphs immediately provides a characterization of diamond-free 
clique-perfect graphs in terms of minimally fordidden induced subgraphs once we prove that within diamond-free graphs clique-perfection and balancedness are equivalent notion. This is accomplished in the following result. Although it can be proved directly using Corollary~\ref{coro:utile}, we prefer to give an almost direct and self-contained proof (an alternative proof can be found in \cite{Safe}).
\begin{theorem}
\label{thm:clique=bal}
Let $G$ be a diamond-free graph. Then $G$ is clique-perfect if and only if it is balanced.
\end{theorem}
\begin{proof}
Since every balanced graph is clique-perfect, the sufficiency condition easily follows.
To prove the necessary condition suppose by contradiction that $G$ is clique-perfect but not balanced. Hence $G$ contains a minimally unbalanced induced subgraph $G_0$. Let $g_0=g(\A_{G_0})$. By Lemma~\ref{lemma:AB} $g_0\geq 5$. By Theorem~\ref{thm:mnbum},
either $\A_{G_0}\cong \C_{g_0}$ or  $\A_{G_0}\cong{\C_{g_0} \brack \K}$ where each row of $\K$ has at least three nonzero entries. 
Since $G$ is clique-perfect it follows that $\A_{G_0}\not\cong \C_{g_0}$. Hence $G_0$ is a multisun with rim $C_0$ of order $g_0$ and inscribed cliques $K^1,\ldots, K^p$ corresponding to the rows of $\K$.
Consider now the parameters $\alpha_c(G_0)$ and $\tau_c(G_0)$. 
First observe that $\lfloor {g_0/2} \rfloor=\alpha_c(C_0)\leq \alpha_c(G_0)$ and $\tau_c(G_0)=\tau_c(C_0)=\lceil {g_0/2} \rceil$ because no edge of $C_0$ is contained in any inscribed clique.

\noindent
Now every clique-independent set of $G_0$ consists of $t$ vertex disjoint cliques among $K^1,\ldots, K^p$ plus some independent edges of $C_0$. Since the latter edges have to be vertex 
disjoint from the chosen $t$ inscribed cliques, it follows that the maximum number of independent edges of $C_0$ occurring in any clique-independent set 
of $G_0$ that contains precisely $t$ inscribed clique, is at most $\lfloor {(g_0-3t)/2} \rfloor$
because each inscribed clique has at least three vertices on $C_0$. Therefore,
\begin{equation}
\label{eq:noKonig}
\alpha_c(G_0)\leq t + \lfloor {(g_0-3t)/2} \rfloor \leq \lfloor {g_0/2} \rfloor.
\end{equation}
It follows that 
$\alpha_c(G_0)=\lfloor {g_0/2} \rfloor$ and so, $\alpha_c(G_0)<\tau_c(G_0)$, contradicting the hypothesis that $G$ is clique-perfect. 
\end{proof}

Theorem~\ref{thm:clique=bal} shows that the recognition problem of clique-perfect diamond-free graphs is polynomial-time solvable
because so is the recognition problem of balanced graphs \cite{Za05}. This answer another question posed by Bonomo et al. in 
\cite{BCDII07}.
Using Theorem~\ref{thm:fund3} we refine the characterization given in the same paper by proving that the odd generalized suns  
forbidden in clique-perfect diamond-free graphs have a very special structure: they are sunoids. 

\begin{theorem}
Let $G$ be a diamond-free graph. Then $G$ is clique-perfect if and only if no induced subgraph of $G$ is an
odd hole or a sunoid.
\end{theorem}

\bibliographystyle{plain}
\bibliography{balanced}

\end{document}